\documentclass[journal]{IEEEtran}
\usepackage{graphicx}

\usepackage{bm}
\usepackage{amsfonts}
\usepackage{amsthm}
\usepackage{amssymb}
\usepackage{MyMnSymbol}
\usepackage{color}
\usepackage{mathrsfs}
\usepackage{flushend}

\usepackage{cite}

\usepackage{booktabs}

\newtheorem{thm}{Theorem}
\newtheorem{lemma}{Lemma}

\ifCLASSINFOpdf
\else
\fi
\usepackage{amsmath}
\usepackage{amsthm}

\usepackage{algorithm}
\usepackage{algorithmic}
\ifCLASSOPTIONcompsoc
 \usepackage[caption=false,font=normalsize,labelfont=sf,textfont=sf]{subfig}
\else
 \usepackage[caption=false,font=footnotesize]{subfig}
\fi
\begin{document}
%
\title{Radar-assisted Predictive Beamforming for Vehicular Links: Communication Served by Sensing}
%
%
%

\author{Fan Liu,~\IEEEmembership{Member,~IEEE,}
        Weijie Yuan,~\IEEEmembership{Member,~IEEE,}\\
        Christos Masouros,~\IEEEmembership{Senior~Member,~IEEE,}
        and~Jinhong Yuan,~\IEEEmembership{Fellow,~IEEE}
\thanks{F. Liu and C. Masouros are with the Department of Electronic and Electrical Engineering, University College London, London, WC1E 7JE, UK (e-mail: fan.liu@ucl.ac.uk, chris.masouros@ieee.org).}
\thanks{W. Yuan and J. Yuan are with the School of Electrical Engineering and Telecommunications, University of New South Wales, Sydney, NSW 2052, Australia (e-mail: weijie.yuan@unsw.edu.au, j.yuan@unsw.edu.au).}
}

\maketitle

\begin{abstract}
In vehicular networks of the future, sensing and communication functionalities will be intertwined. In this paper, we investigate a radar-assisted predictive beamforming design for vehicle-to-infrastructure (V2I) communication by exploiting the dual-functional radar-communication (DFRC) technique. Aiming for realizing joint sensing and communication functionalities at road side units (RSUs), we present a novel extended Kalman filtering (EKF) framework to track and predict kinematic parameters of each vehicle. By exploiting the radar functionality of the RSU we show that the communication beam tracking overheads can be drastically reduced. To improve the sensing accuracy while guaranteeing the downlink communication sum-rate, we further propose a power allocation scheme for multiple vehicles. Numerical results have shown that the proposed DFRC based beam tracking approach significantly outperforms the communication-only feedback based technique in the tracking performance. Furthermore, the designed power allocation method is able to achieve a favorable performance trade-off between sensing and communication.
\end{abstract}


%
\IEEEpeerreviewmaketitle

\section{Introduction}
%
%
%
%
\IEEEPARstart{V}{ehicle}-to-everything (V2X) communication will play an important role in the next-generation autonomous vehicles, requiring low-latency Gbps data transmission \cite{8246850}. In addition to wireless communication services, the sensing capability is also essential in the future V2X network, as it should be able to provide robust obstacle detection and high-accuracy localization services on the order of a centimeter \cite{8246850,8306879}. At the time of writing, vehicular localization and networking schemes are mainly built upon global navigation satellite-based systems (GNSS), or default standards such as dedicated short-range communication (DSRC)\cite{5888501} and the D2D mode of LTE-A \cite{7786130}. While these techniques do offer basic V2X functionalities, they are unlikely to fulfill the demanding requirements mentioned above. For instance, the 4G cellular system provides Mbps communication services, with coarse positioning capability at an accuracy of tens of meters, and at a latency that often in excess of 1s \cite{8246850}. With the assistance of dedicated base stations, the real-time kinematic GNSS can reduce the positioning error to the centimeter-level, at the price of considerable latency and low refresh rate \cite{8246850}.
\\\indent To tackle the above issues in the safety-critical vehicular applications, the forthcoming 5G technology, which exploits both the massive multi-input-multi-output (mMIMO) antenna array and the mmWave spectrum, is envisioned as a promising solution \cite{6515173,7400949}. The large bandwidth available in the mmWave band not only offers the advantage of higher data rate, but also significantly improves the resolution for range estimation. On the other hand, the mMIMO array is able to compensate for the path-loss imposed on the mmWave signals by formulating ``pencil-like" beams that point to the directions of users, which also enhances the angular resolution from a localization perspective. As per the above localization and networking requirements, joint sensing and communication designs naturally arise in the V2X scenarios. Having a single device providing both sensing and communication functionalities may significantly reduce the hardware complexity associated with sensors mounted on the vehicles or road infrastructures, while improving the overall performance. To this end, research efforts towards dual-functional radar-communication (DFRC) systems are well underway.
\\\indent Early contributions for DFRC mainly focus on integrating radar and communication signals on the temporal or the frequency domains, where a classic design is to employ a chirp signal, a commonly used radar probing waveform, as the carrier for communication messages. Such examples can be found in \cite{roberton2003integrated,saddik2007ultra}, where ``0" and ``1" are represented by down- and up-chirp signals. Alternatively, pseudo-random codes, which are also able to spread the spectrum, can be used both as radar probing signals and information carriers \cite{jamil2008integrated}. Orthogonal Frequency Division Multiplexing (OFDM) waveform, which has been extensively employed for communications, has been recently regarded as a promising solution for realizing DFRC \cite{sturm2011waveform}. This is because by using OFDM communication signals for radar sensing, the conventionally coupled range and Doppler estimators become independent to each other, and the impact of the random communication data can be simply mitigated by element-wise division \cite{sturm2011waveform}. To further enhance the radar performance, one can also replace the sinusoidal carrier of the OFDM with the chirp signal \cite{7485314}. Accordingly, the fractional Fourier Transform (FrFT) will be employed for signal processing instead of the Fast Fourier Transform (FFT) \cite{7485314}.
\\\indent With the development of multi-antenna technology, the spatial signal processing for DFRC has been explored by more recent treatises. In \cite{7347464}, the authors proposed to detect target using the mainlobe of the MIMO radar, while transmitting useful information to the communication receivers using the sidelobes. Such techniques typically employ a inter-pulse modulation, leading to a low data rate that is tied to the pulse repetition frequency (PRF) of the radar, e.g., kbps. Moreover, the sidelobe based DFRC schemes can work only when a Line-of-Sight (LoS) path exists, which limits its applications in more realistic Non-LoS (NLOS) scenarios. Therefore, the authors of \cite{8288677} have proposed a novel DFRC beamforming design for NLOS communication scenarios, which allows the use of intra-pulse modulation, and hence improves the data rate. Aiming for adapting the nonlinearity of the power amplifiers in both radar and communication systems, the work of \cite{8386661} further considered a constant modulus (CM) DFRC waveform design, which can be efficiently obtained by a sophisticated branch-and-bound (BnB) algorithm.
\\\indent It should be highlighted that in the above works, the DFRC schemes are mainly conceived for lower frequency band, e.g., sub-6 GHz, and are thus difficult to be extended to the V2X applications that operate in the mmWave band. While the work in \cite{Joint_TCOM} proposed a novel framework for mmWave DFRC transceivers, it is not tailored for vehicular communications. More relevant to this work, a radar-aided beam alignment method has been designed in \cite{7888145} for mmWave vehicle-to-infrastructure (V2I) communications. Unlike the DFRC schemes considered above, the authors of \cite{7888145} proposed to deploy an extra radar device in addition to the communication system, which inevitably results in high hardware costs. In view of this, a mmWave DFRC system has been proposed in \cite{8642926} for compromising the bi-static automotive radar and vehicle-to-vehicle (V2V) communications. However, this work does not address the issue of beam tracking under vehicular scenarios with high mobility, which is key to guaranteeing the quality-of-service (QoS) for V2X communications.
\\\indent Existing works on mmWave beam tracking are generally based on the communication-only protocols \cite{8809900,7999215}. A typical beam tracking process requires the transmitter to send pilots to the receiver, according to which the receiver estimates the angle and feeds it back to the transmitter. Note that for high-mobility communication scenarios, it is not sufficient to only track the beam. In fact, the transmitter should have the capability to \emph{predict} the beam, in order to meet the critical latency requirement. To achieve this, the state-of-the-art approaches have exploited Kalman filtering both as an estimator and a predictor for mmWave beam tracking based on the feedback protocol mentioned above \cite{7905941,8025577,8851151,8830375}. Nevertheless, these techniques typically utilize only a small number of pilots for beam tracking (in the extreme case only a single pilot is transmitted \cite{7905941,8851151,8830375}), leading to a limited matched filtering gain for angle estimation, and may thus suffer from serious performance-loss in the vehicular applications that require high-accuracy localization services. Moreover, inserting pilot symbols into the communication block results in a communication overhead, which degrades the transmission rate of the useful information.
\\\indent In this paper, we offer a solution that removes the feedback loop and the resulting signalling overhead for vehicular beamtracking. Our work is based on a novel predictive beamforming design for the multi-user V2I network by relying on the DFRC techniques. To be specific, we employ DFRC signals in the downlink transmission, where the echo signals reflected by the vehicles are exploited for tracking and localization. In other words, the whole downlink block is jointly used both as radar sensing signals and communication data symbols. As a result, no downlink pilots are needed, and the matched-filtering operation for the echo signals would bring significant gain in the signal-to-noise ratio (SNR). Following the spirit of joint sensing and communication, we consider an extended Kalman filtering (EKF) scheme for tracking and predicting the angle, the distance and the velocity under the kinematic model of each vehicle. Note here that while the angle information is used for downlink beamforming, both the angle and the distance need to be tracked for the purpose of localizing the vehicles. To further improve the localization accuracy while guaranteeing the communication QoS, we propose a novel power allocation method for multiple vehicles, which is able to minimize the estimation errors while ensuring the downlink sum-rate constraint. For clarity, we summarize our contribution as follows:
\begin{itemize}
  \item We propose a novel DFRC based framework for joint sensing and communication in the V2I network, which requires no downlink pilots in contrast to its communication-only feedback based counterpart;
  \item We propose a novel EKF method that is able to accurately track and predict the state of the vehicle;
  \item We propose a novel power allocation scheme aiming at improving the sensing performance while guaranteeing the downlink sum-rate for multiple vehicles, which achieves a favorable performance tradeoff between sensing and communication.
\end{itemize}

The remainder of this paper is organized as follows, Section II introduces the system model, Section III describes the proposed EKF approach, Section IV proposes the power allocation design for multiple vehicles, Section V provides the numerical results, and finally Section VI concludes the paper.
\\\indent {\emph{Notation}}: Unless otherwise specified, matrices are denoted by bold uppercase letters (i.e., $\mathbf{A}$), vectors are represented by bold lowercase letters (i.e., $\mathbf{f}$), and scalars are denoted by normal font (i.e., $\theta$). Subscripts indicate the indexes of the time-slot and the vehicle, (i.e., $\theta_{k,n}$ denotes the angle of the $k$th vehicle at the $n$th epoch).

\section{System Model}
We consider a mmWave mMIMO RSU with a uniform linear array (ULA), which serves $K$ vehicles on the road as depicted in Fig. 1. To communicate with the RSU, each vehicle is also equipped with an MIMO array at both sides of the body. For notational simplicity, and without loss of generality, we assume that the vehicles are driving along a straight road that is parallel to the antenna array of the RSU, and that the RSU communicates with each vehicle via a LoS channel. The discussion of NLoS channels is designated to our future work. In what follows, we will firstly introduce the general framework, and then the detailed signal model.
\\\indent \emph{Remark 1:} Note that the ULA of the RSU can be adjusted to be paralleled to the road, where small mismatches are allowed. In fact, alternative relative directions can be straightforwardly accommodated by adding a fixed offset to the tracked angles. We note here that this offset can be easily calibrated since it is fixed and is known to the RSU. As a result, our proposed techniques can be applied without any changes.
\begin{figure}[!t]
    \centering
    \includegraphics[width=0.7\columnwidth]{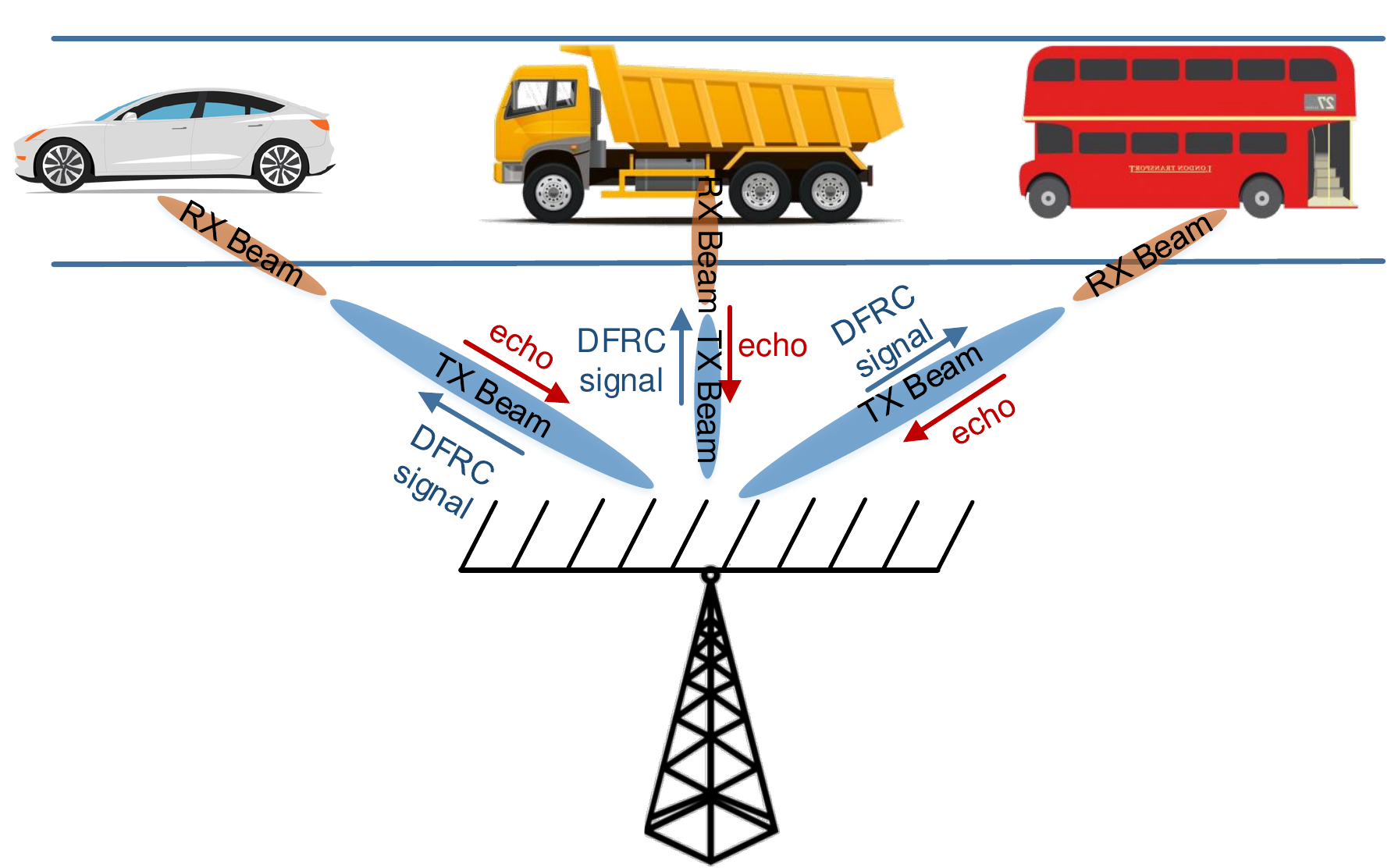}
    \caption{V2I scenario model.}
    \label{fig:1}
\end{figure}
\subsection{The General Framework}
To establish reliable communication links, the RSU needs to acquire the accurate information for the azimuth angle of the vehicles. On the other hand, it is also necessary to have the knowledge of RSU's relative angle at each vehicle. By doing so, both the RSU and the vehicle could use their antenna arrays to formulate narrow beams that could accurately point to the directions of each other. Conventionally, the beam alignment is done in a \emph{scanning} mechanism, i.e., periodically transmitting and receiving pilots at all the possible beams, and thus to search for the beam pair that gives the strongest path gain \cite{7914742}. Nevertheless, such a scheme will inevitably cause considerable latency and communication overhead since both pilots and feedbacks between TX and RX are required. Given the high mobility of the vehicles, it is important to efficiently track or even predict the variation of all the angular parameters. Moreover, the need for high-accuracy localization requires the knowledge of the distance of each vehicle in addition to its angular counterpart. To accomplish both sensing and communication tasks, we propose in the following a framework based on the DFRC technique.
\\\indent Let us denote the angle, the distance and the velocity of the \emph{k}th vehicle relative to the RSU's array as $\theta_k\left(t\right)$, $d_k\left(t\right)$, $v_k\left(t\right)$, respectively. Further, the angle of the RSU relative to the $k$th vehicle is denoted as $\phi_k\left(t\right)$. Note that all the parameters are functions of time $t \in \left[0,T\right]$, with $T$ being the maximum time duration of interest. It then follows that $\phi_k\left(t\right) = \theta_k\left(t\right), \forall k$, given the parallel driving directions of the vehicles relative to the RSU's antenna array. We therefore omit $\phi_k$ in the remainder of the paper. For notational convenience, we discretize the time period $T$ into several small time-slots with a length of $\Delta T$, and denote $\theta_{k,n}$, $d_{k,n}$ and $v_{k,n}$ as the motion parameters at the \emph{n}th epoch for each vehicle. Following the standard assumption in the literature \cite{7905941,8025577,8851151,8830375}, we assume that the motion parameters keep constant within $\Delta T$.
\\\indent \emph{1) Initial Estimation}
\\\indent Our proposed scheme is initialized by letting the RSU estimate the parameters of the vehicles that enter into the coverage of interest. In this stage, the RSU can either act as a pure mono-static radar, which infers the initial vehicle parameters $\theta_{k,0}$, $d_{k,0}$ and $v_{k,0}$ from the reflected echoes, or to obtain these estimates simply via conventional uplink training. Here we note that while the RSU is only able to attain the radial velocity $v_{k,n}^{R}$ by estimating the Doppler frequency, it can infer the overall velocity as $v_{k,n} = v_{k,n}^{R}/\cos\theta_{k,n}$.
\\\indent\emph{2) State Prediction}
\\\indent With the estimates of the motion parameters $\hat \theta_{k,n-1}$, $\hat{d}_{k,n-1}$ and ${\hat v}_{k,n-1}$ at the $\left(n-1\right)$th epoch, the RSU performs one- and two-step predictions of the angle parameters, respectively. For the purpose of sensing, the RSU will also need to perform one-step prediction for other motion parameters, i.e., distance and velocity. At the $n$th epoch, the RSU formulates \emph{K} transmit beams towards \emph{K} vehicles by using the one-step predictions ${\hat \theta _{k,{n\left| {n - 1} \right.}}},\forall k$. In each of the beams, the RSU will send a joint radar-communication signal that contains the information of the two-step predictions ${{\hat \theta }_{k,n + 1\left| {n - 1} \right.}},\forall k$. Once the vehicles receive the information, they will correspondingly formulate receive beams at the $\left(n+1\right)$th epoch based on the predicted angles. The reason for using the two-step prediction at the vehicle is that the one-step predicted angle ${\hat \theta _{k,{n\left| {n - 1} \right.}}}$ would be outdated at the $\left(n+1\right)$th epoch. Note that the predictions are performed by using the kinematic equations of the vehicles. The transmit beams of the RSU and the receive beams of the vehicles will be aligned with each other if the estimation and prediction are sufficiently accurate.
\\\indent\emph{3) Vehicle Tracking}
\\\indent At the $n$th epoch, the signal transmitted by the RSU via each beam is partially reflected by the body of the corresponding vehicle, and is also partially received by the vehicle's antenna array. As discussed above, for each vehicle, the data sequence received contains the predicted angular information for the $\left(n+1\right)$th epoch, which will be exploited for receive beamforming at the vehicles. On the other hand, the RSU receives the echoes reflected by all vehicles, and estimates $\theta_{k,n},v_{k,n}$ and $d_{k,n}$, which are used to refine the predicted parameters at the $n$th epoch. The refined state parameters are then used as the inputs of the predictor for the $\left(n+1\right)$th and $\left(n+2\right)$th epoches at the RSU.
\\\indent For clarity, we summarize the above procedure in Fig. 1 and Fig. 2. It can be observed that by iteratively performing beam prediction and beam tracking, the RSU is able to serve multiple vehicles simultaneously. Moreoever, with the aid of the radar functionality built in the RSU, one can avoid frequent feedbacks between the RSU and vehicles. This is evidently shown in Fig. 2 that the uplink feedback from the vehicles to the RSU are replaced by the echo signal. In this sense, the beam information can be extracted by the echo signal, and all the uplink resources can be used to transfer useful data rather than the feedback information.

\subsection{Signal Model}
Based on the above discussion, it is clear that the dominating complexity in the signal processing is at the RSU's side. Moreover, the initial estimation can be simply done by conventional radar signal processing or uplink beam training. Given the aforementioned reasons, we will focus on the prediction and tracking stages at the RSU's side. In this subsection, we develop the measurement model at the RSU by using radar signal processing techniques.
\\\indent \emph{1) Radar Signal Model}
\\\indent Let us denote the $K$ downlink DFRC streams transmitted at the $n$th epoch and time $t$ as ${{\mathbf{s}}_n}\left( t \right) = {\left[ {{s_{1,n}}\left( t \right), \ldots ,{s_{K,n}}\left( t \right)} \right]^T} \in {\mathbb{C}^{K \times 1}}$. The transmitted signal can be expressed as
\begin{equation}\label{eq1}
  {{\mathbf{\tilde s}}}_n\left(t\right) = {\mathbf{F}}_n{\mathbf{s}}_n\left(t\right) \in \mathbb{C}^{N_t \times 1},
\end{equation}
where ${\mathbf{F}}_n \in \mathbb{C}^{N_t \times K}$ is the transmit beamforming matrix, with $N_t$ being the number of transmit antennas. Accordingly, the reflected echoes received at the RSU can be given in the form
\begin{equation}\label{eq2}
\begin{gathered}
  {{\mathbf{r}}_n}\left( t \right) = \hfill \\
  \kappa\sum\limits_{k = 1}^K \sqrt {p_{k,n}}{{\beta _{k,n}}{e^{j{2\pi \mu_{k,n}}t}}{\mathbf{b}}\left( {{\theta _{k,n}}} \right){{\mathbf{a}}^H}\left( {{\theta _{k,n}}} \right)} {\mathbf{\tilde s}}_n\left( {t - {\tau _{k,n}}} \right) \hfill \\
  + {{\mathbf{z}}_r}\left( t \right), \hfill \\
\end{gathered}
\end{equation}
where $p_{k,n}$ is the transmit power at the $k$th beam and the $n$th epoch, $\kappa = \sqrt{N_tN_r}$ is the array gain factor, with $N_r$ being the number of receive antennas, $\mathbf{z}_r\left( t \right) \in \mathbb{C}^{N_r \times 1}$ represents the complex additive white Gaussian noise with zero mean and variance of $\sigma^2$, $\beta_{k,n}$, $\mu_{k,n}$ and $\tau_{k,n}$ denote the reflection coefficient, the Doppler frequency and the time-delay for the $k$th vehicle. Given the distance $d_{k,n}$, the reflection coefficient for the $k$th vehicle can be expressed as
\begin{equation}\label{eq3}
  {\beta _{k,n}} = {\varepsilon _{k,n}}{\left( {2{d_{k,n}}} \right)^{ - 1}},
\end{equation}
where $\varepsilon _{k,n}$ is the complex radar cross-section (RCS) of the $k$th vehicle at the $n$th epoch. We assume that the RCS of the vehicle keeps constant during the period $T$, i.e., ${\varepsilon _{k,n}} = {\varepsilon _{k,n-1}} = ...= {\varepsilon _{k,0}}, \forall n$, which corresponds to a Swerling I target \cite{richards2005fundamentals}. In (\ref{eq2}), ${\mathbf{a}}\left( {{\theta}} \right)$ and ${\mathbf{b}}\left( {{\theta}} \right)$ are transmit and receive steering vectors of the antenna array of the RSU, which are expressed as
\begin{equation}\label{eq4}
{\mathbf{a}}\left( \theta  \right) = \sqrt{\frac{1}{N_t}}{\left[ {1,{e^{-j\pi \cos \theta }},...,{e^{-j\pi \left( {{N_t} - 1} \right)\cos \theta }}} \right]^T},
\end{equation}
\begin{equation}\label{eq5}
{\mathbf{b}}\left( \theta  \right) = \sqrt{\frac{1}{N_r}}{\left[ {1,{e^{-j\pi \cos \theta }},...,{e^{-j\pi \left( {{N_r} - 1} \right)\cos \theta }}} \right]^T},
\end{equation}
where we assume half-wavelength antenna spacing for the ULA.
\\\indent The beamforming matrix $\mathbf{F}_n$ is designed based on the prediction of the angles. The $k$th column of $\mathbf{F}_n$ is given as
\begin{equation}\label{eq6}
  {{\mathbf{f}}_{k,n}} = \mathbf{a}\left( {{{\hat \theta }_{k,{n\left| {n - 1} \right.}}}} \right),\forall k,
\end{equation}
where ${{{\hat \theta }_{k,{n\left| {n - 1} \right.}}}} $ is the one-step predicted angle for the $k$th vehicle at the $n$th epoch. By employing the design in (\ref{eq6}), the RSU will formulate $K$ beams towards the predicted directions to track the vehicles.
\\\indent \emph{2) Radar Measurement Model}
\\\indent By using the massive MIMO array at the RSU, the beams formulated will be sufficiently narrow, such that the inter-beam interference can be omitted. This follows from the established mMIMO theory, and can be mathematically expressed by the following lemma.
\begin{lemma}
  For uniform linear array, we have $\left| {{{\mathbf{a}}^H}\left( {{\theta}} \right){\mathbf{a}}\left( {{\phi}} \right)} \right| \to 0, \forall \theta \ne \phi$, ${N_t} \to \infty$.
\end{lemma}
\renewcommand{\qedsymbol}{$\blacksquare$}
\begin{proof}
  See \cite{ngo2015massive}.
\end{proof}
Lemma 1 suggests that the steering vectors are asymptotically orthogonal to each other under the massive MIMO regime. As a result, the reflected echoes from different vehicles will not interfere with each other, and the RSU can thus process each echo signal individually. For the $k$th vehicle, the received echo at its associated beam is given by
\begin{equation}\label{eq7}
\begin{gathered}
{{\mathbf{r}}_{k,n}}\left( t \right) = \hfill \\
\kappa\sqrt {p_{k,n}}{\beta _{k,n}}{e^{j{2\pi \mu_{k,n}}t}}{\mathbf{b}}\left( {{\theta _{k,n}}} \right){{\mathbf{a}}^H}\left( {{\theta _{k,n}}} \right){{\mathbf{f}}_{k,n}}{{s}}_{k,n}\left( {t - {\tau _{k,n}}} \right) \hfill \\
+ \mathbf{z}_{k,n}\left( t \right). \hfill \\
\end{gathered}
\end{equation}
In (\ref{eq7}), the transmit signal-to-noise ratio (SNR) is defined as $ \frac{{{p_{k,n}}}}{{{\sigma ^2}}}$. By matched-filtering (\ref{eq7}) with a delayed and Doppler-shifted version of ${{s}}_{k,n}\left( {t} \right)$, one can estimate the delay $\tau_{k,n}$ and the Doppler frequency $\mu_{k,n}$. Compensating (\ref{eq7}) using these estimates yields the measurement model for the angle $\theta_{k,n}$ and the reflection coefficient $\beta_{k,n}$ as
\begin{equation}\label{eq8}
\begin{gathered}
  {{\mathbf{\tilde r}}_{k,n}} = \kappa{\beta _{k,n}}{\mathbf{b}}\left( {{\theta _{k,n}}} \right){{\mathbf{a}}^H}\left( {{\theta _{k,n}}} \right){{\mathbf{f}}_{k,n}} + {{\mathbf{z}}_{\theta}} \hfill \\
   = \kappa{\beta _{k,n}}{\mathbf{b}}\left( {{\theta _{k,n}}} \right){{\mathbf{a}}^H}\left( {{\theta _{k,n}}} \right){\mathbf{a}}\left( {{{\hat \theta }_{k,n\left| {n - 1} \right.}}} \right) + {{\mathbf{z}}_{\theta}}, \hfill \\
\end{gathered}
\end{equation}
where ${\mathbf{z}}_{\theta}$ denotes the measurement noise normalized by the transmit power $p_{k,n}$ and the matched-filtering gain $G$, with zero mean and variance of $\sigma_1^2$. Note here that $G$ is the SNR gain brought by the matched-filtering operation, which typically equals to the energy of $s_{k,n}\left(t\right)$. Furthermore, the measurement models of the distance $d_{k,n}$ and the velocity $v_{k,n}$ are given as
\begin{equation}\label{eq9}
  {\tau _{k,n}} = \frac{{2{d_{k,n}}}}{c} + {z_\tau },
\end{equation}
\begin{equation}\label{eq10}
  {\mu_{k,n}} = \frac{{2{v_{k,n}}\cos {\theta _{k,n}}{f_c}}}{c} + {z_f},
\end{equation}
where $f_c$ and $c$ represent the carrier frequency and the speed of light, respectively, ${z_\tau }$ and $z_f$ denote the measurement Gaussian noise with zero mean and variance of $\sigma_2^2$ and $\sigma_3^2$, respectively. Note that the round-trip is twice the distance from the RSU to the vehicle, and the Doppler frequency relies on the radial velocity ${v_{k,n}}\cos {\theta _{k,n}}$. Moreover, we remark here that the variances of the measurement noises are inversely proportional to the receive SNR of (\ref{eq7}) \cite{kay1998fundamentals}, i.e.,
\begin{equation}
  \sigma _1^2 \propto \frac{{{\sigma ^2}}}{{G{p_{k,n}}}},\sigma _i^2 \propto \frac{{{\sigma ^2}}}{{G{\kappa ^2}{{\left| {{\beta _{k,n}}} \right|}^2}{\left|\delta _{k,n}\right|}^2{p_{k,n}}}},i = 2,3,
\end{equation}
where ${\delta _{k,n}} = {{\mathbf{a}}^H}\left( {{\theta _{k,n}}} \right){\mathbf{a}}\left( {{{\hat \theta }_{k,n\left| {n - 1} \right.}}} \right)$ represents the beamforming gain factor, whose modulus equals to 1 if the predicted angle perfectly matches the real angle, and is less than 1 otherwise. For convenience, we assume that
\begin{equation}
  \sigma _1^2 = \frac{{{a_1^2}{\sigma ^2}}}{{G{p_{k,n}}}},\sigma _i^2 = \frac{{{a_i^2}{\sigma ^2}}}{{G{\kappa ^2}{{\left| {{\beta _{k,n}}} \right|}^2}{\left|\delta _{k,n}\right|}^2{p_{k,n}}}},i = 2,3.
\end{equation}
Note that $\sigma_2^2$ and $\sigma_3^2$ are determined by the transmit power $p_{k,n}$, the matched filtering gain $G$, the array gain $\kappa$, the beamforming gain ${\delta _{k,n}}$ as well as the strength of the reflected signal. Nevertheless, $\sigma _1^2$ is only determined by the transmit power $p_{k,n}$ and the matched filtering gain $G$, since $\kappa$, $\beta_{k,n}$ and ${\delta _{k,n}}$ are already contained in (\ref{eq8}). Finally, $a_i, i= 1,2,3$ are constants related to the system configuration, signal designs as well as the specific signal processing algorithms.
\\\indent \emph{3) Communication Model}
\\\indent As shown in Fig. 2, at the $n$th epoch, the $k$th vehicle receives the signal from the RSU by using a receive beamformer $\mathbf{w}_{k,n}$, yielding
\begin{equation}\label{eq11}
\begin{gathered}
  {c_{k,n}}\left( t \right) = \hfill \\
    {\tilde\kappa}\sqrt{p_{k,n}}{\alpha _{k,n}}{\mathbf{w}}_{k,n}^H{\mathbf{u}}\left( {{\theta _{k,n}}} \right){{\mathbf{a}}^H}\left( {{\theta _{k,n}}} \right){{\mathbf{f}}_{k,n}}{s_{k,n}}\left( t \right) + {z_c}\left( t \right), \hfill \\
\end{gathered}
\end{equation}
where ${{z}_c}\left( t \right)$ is the zero-mean Gaussian noise with variance ${\sigma_C^2}$, ${s}_{k,n}\left(t\right)$ denotes the DFRC stream transmitted from the RSU to the $k$th vehicle, ${\alpha _{k,n}}$ denotes the communication channel coefficient, which is different from the radar reflection coefficient ${\beta _{k,n }}$, ${\mathbf{u}}\left(\theta\right)$ represents the steering vector of the vehicle's antenna array, and is similarly defined as in (\ref{eq3}) and (\ref{eq4}) with $M_k$ antennas. For notational simplicity, let us assume $M_1 = ... = M_K = M$. Again, ${\tilde\kappa} = \sqrt{N_tM}$ is the array gain factor. Note that the inter-vehicle interference vanishes thanks to the narrow beam generated by the mMIMO array.
\\\indent As discussed in the above, the receive beamformer should be formulated based on the two-step prediction of the angle parameter, since the one-step predicted information would be outdated for receive beamforming at the vehicle. This is expressed as
\begin{equation}\label{eq12}
{{\mathbf{w}}_{k,n}} = {\mathbf{u}}\left( {{{\hat\theta} _{k,n\left| {n - 2} \right.}}} \right).
\end{equation}
Assume that the DFRC stream ${s}_{k,n}\left(t\right)$ has a unit power, then the receive signal-to-noise ratio (SNR) for the $k$th vehicle at the $n$th epoch is obtained as
\begin{equation}\label{eq13}
\begin{gathered}
  {\operatorname{SNR} _{k,n}} = \frac{{p_{k,n}}{{{\left| {\tilde\kappa}{{\alpha _{k,n}}{\mathbf{w}}_{k,n}^H{\mathbf{u}}\left( {{\theta _{k,n}}} \right){{\mathbf{a}}^H}\left( {{\theta _{k,n}}} \right){{\mathbf{f}}_{k,n}}} \right|}^2}}}{{{{\sigma_C^2}}}} \hfill \\
   = {{{p_{k,n}}{\rho_{k,n}}}},\hfill \\
\end{gathered}
\end{equation}
where
\begin{equation}\label{eq14}
{\rho _{k,n}} = {{{{\left| \begin{gathered}
  {\tilde\kappa}{\alpha _{k,n}}{{\mathbf{u}}^H}\left( {{{\hat \theta }_{k,n\left| {n - 2} \right.}}} \right){\mathbf{u}}\left( {{\theta _{k,n}}} \right) \hfill \\
   \cdot {{\mathbf{a}}^H}\left( {{\theta _{k,n }}} \right){\mathbf{a}}\left( {{{\hat \theta }_{k,n\left| n-1 \right.}}} \right) \hfill \\
\end{gathered}  \right|}^2}} \mathord{\left/
 {\vphantom {{{{\left| \begin{gathered}
  {\alpha _{k,n}}{{\mathbf{u}}^H}\left( {{{\hat \theta }_{k,n\left| {n - 2} \right.}}} \right){\mathbf{u}}\left( {{\theta _{k,n}}} \right) \hfill \\
   \cdot {{\mathbf{a}}^H}\left( {{\theta _{k,n}}} \right){\mathbf{a}}\left( {{{\hat \theta }_{k,n\left| n-1 \right.}}} \right) \hfill \\
\end{gathered}  \right|}^2}} {{{\sigma_C^2}}}}} \right.
 \kern-\nulldelimiterspace} {{{\sigma_C^2}}}}.
\end{equation}
The achievable sum-rate of all the $K$ vehicles is thus given as
\begin{equation}\label{eq15}
  {R_{n}} = \sum\limits_{k = 1}^K {{{\log }_2}\left( {1 + {\operatorname{SNR} _{k,n }}} \right)}
   = \sum\limits_{k = 1}^K {{{\log }_2}\left( {1 + {{{p_{k,n}}{\rho _{k,n}}}}} \right)}.
\end{equation}
Following the standard assumption in the literature, the LoS channel coefficient $\alpha_{k,n}$ is given as \cite{8675440}
\begin{equation}\label{eq16}
  {\alpha _{k,n}} = {\tilde \alpha}d_{k,n}^{ - 1}{e^{j\frac{{2\pi }}{\lambda }{d_{k,n}}}} = {{\tilde \alpha}}d_{k,n}^{ - 1}{e^{j\frac{{2\pi {f_c}}}{c}{d_{k,n}}}},
\end{equation}
where ${\tilde \alpha}d_{k,n}^{ - 1}$ is the path-loss of the channel, with ${\tilde \alpha}$ being the channel power gain at the reference distance $d_0 = 1\text{m}$, $\frac{{2\pi }}{\lambda }{d_{k,n}}$ is the phase of the LoS channel, with $\lambda = \frac{{{f_c}}}{c}$ being the wavelength of the signal. The reference power gain factor ${\tilde \alpha}$ is assumed to be known to the RSU. Therefore, to estimate ${\alpha _{k,n}}$ is equivalent to estimating $d_{k,n}$.
\begin{figure}[!t]
    \centering
    \includegraphics[width=0.7\columnwidth]{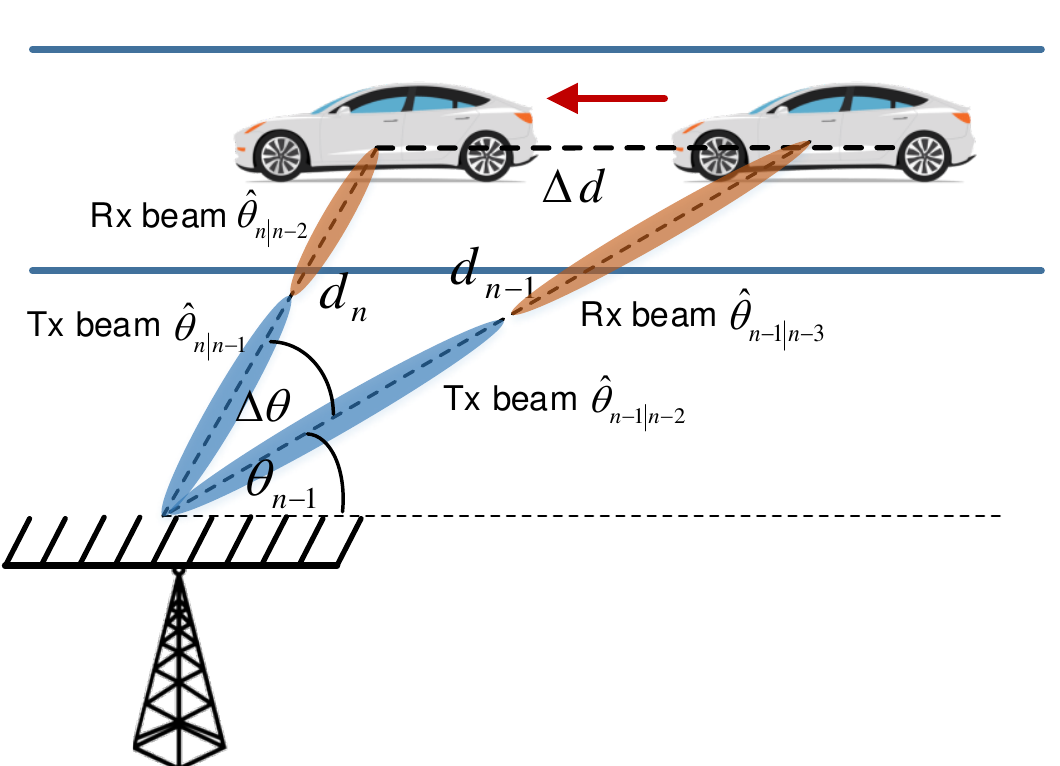}
    \caption{V2I state evolution model.}
    \label{fig:2}
\end{figure}
\subsection{State Evolution Model}
Our goal is to track the variation of the angles and distances of $K$ vehicles by processing the measured signals in (\ref{eq8}), (\ref{eq9}) and (\ref{eq10}), which are determined by the kinematic equations of the vehicles. According to the above discussion, the RSU can process the echo signal reflected by each vehicle individually, given the orthogonality among $K$ vehicles from Lemma 1. We therefore focus on a simple scenario where the RSU is serving a single vehicle, and omit the subscript $k$ for notational convenience. At the $n$th epoch, the angle, the distance, the velocity and the reflection coefficient of the vehicle are denoted as $\theta_n$, $d_n$, $v_n$ and $\beta_n$ for simplicity. In the following, we derive the state evolution model of the vehicle. First of all, based on the geometric relations shown in Fig. 2, we have
\begin{equation}\label{eq17}
\left\{ \begin{gathered}
  d_n^2 = d_{n - 1}^2 + \Delta {d^2} - 2{d_{n - 1}}\Delta d\cos {\theta _{n - 1}}, \hfill \\
  \frac{{\Delta d}}{{\sin \Delta \theta }} = \frac{{{d_n}}}{{\sin {\theta _{n - 1}}}}, \hfill \\
\end{gathered}  \right.
\end{equation}
where $\Delta d = {v_{n-1}}\Delta T$, $\Delta \theta = \theta_n - \theta_{n-1}$. It is challenging to analyze the evolution model directly given the high non-linear nature of (\ref{eq17}). Therefore, we propose an accurate approximation of (\ref{eq17}) in the following. Note that one can rewrite the first equation in (\ref{eq17}) as
\begin{equation}\label{eq18}
\begin{gathered}
  d_n^2 - d_{n - 1}^2 = \left( {{d_n} + {d_{n - 1}}} \right)\left( {{d_n} - {d_{n - 1}}} \right) \hfill \\
   = \Delta {d^2} - 2{d_{n - 1}}\Delta d\cos {\theta _{n - 1}}. \hfill \\
\end{gathered}
\end{equation}
Hence we have
\begin{equation}\label{eq19}
\begin{gathered}
  {d_n} - {d_{n - 1}} = \frac{{\Delta {d^2} - 2{d_{n - 1}}\Delta d\cos {\theta _{n - 1}}}}{{{d_n} + {d_{n - 1}}}} \hfill \\
   \approx \frac{{\Delta {d^2} - 2{d_{n - 1}}\Delta d\cos {\theta _{n - 1}}}}{{2{d_{n - 1}}}} = \Delta d\left( {\frac{{\Delta d}}{{2{d_{n - 1}}}} - \cos {\theta _{n - 1}}} \right), \hfill \\
\end{gathered}
\end{equation}
where the approximation in (\ref{eq19}) is based on the fact that the position of the vehicle will not change too much in such a short time duration. To verify this, let us consider a simple example where a vehicle travels at a speed of $v = 54\text{km/h} = 15\text{m/s}$, within a short period $\Delta T = 10\text{ms}$, the traveling distance is $\Delta d = 0.15\text{m}$, which can be negligible comparing to the typical distance from the vehicle to the RSU, e.g., tens or hundreds of meters. Accordingly, ${\frac{{\Delta d}}{{2{d_{n - 1}}}}}$ becomes negligible in (\ref{eq19}), which yields
\begin{equation}\label{eq20}
  {d_n} \approx {d_{n - 1}} - \Delta d\cos {\theta _{n - 1}} = {d_{n - 1}} - {v_{n - 1}}\Delta T\cos {\theta _{n - 1}}.
\end{equation}
We further note that when $\Delta \theta$ is small, one can approximate $\sin\Delta\theta \approx \Delta\theta$. This gives us
\begin{equation}\label{eq21}
  \Delta \theta  \approx \sin \Delta \theta  = \frac{{\Delta d\sin {\theta _{n - 1}}}}{{{d_n}}}.
\end{equation}
Substituting (\ref{eq20}) into (\ref{eq21}) leads to
\begin{equation}\label{eq22}
  \Delta \theta  \approx \frac{{\Delta d\sin {\theta _{n - 1}}}}{{{d_{n - 1}} - \Delta d\cos {\theta _{n - 1}}}} = \frac{{\tan {\theta _{n - 1}}}}{{\frac{{{d_{n - 1}}}}{{\Delta d\cos {\theta _{n - 1}}}} - 1}}.
\end{equation}
Again, by using the property $\Delta d \ll {d_{n - 1}}$, (\ref{eq22}) can be further simplified as
\begin{equation}\label{eq23}
  \Delta \theta  \approx \frac{{\Delta d\cos {\theta _{n - 1}}\tan {\theta _{n - 1}}}}{{{d_{n - 1}}}} = \frac{{\Delta d\sin {\theta _{n - 1}}}}{{{d_{n - 1}}}}.
\end{equation}
Hence we have
\begin{equation}\label{eq24}
{\theta _n} \approx {\theta _{n - 1}} + \frac{{\Delta d\sin {\theta _{n - 1}}}}{{{d_{n - 1}}}} = {\theta _{n - 1}} + d_{n - 1}^{ - 1}{v_{n - 1}}\Delta T\sin {\theta _{n - 1}}.
\end{equation}
By assuming that the vehicle is moving at an approximately constant speed, we have
\begin{equation}\label{eq25}
  v_n \approx v_{n-1}.
\end{equation}
We then analyze the evolution of the reflection coefficient $\beta$. Recalling the definition in (\ref{eq3}), it is straightforward to see
\begin{equation}\label{eq26}
  {\beta _n} = {\varepsilon _n}{\left( {2{d_n}} \right)^{ - 1}},{\beta _{n - 1}} = {\varepsilon _{n - 1}}{\left( {2{d_{n - 1}}} \right)^{ - 1}},
\end{equation}
where ${\varepsilon _n}$ and ${\varepsilon _{n-1}}$ denote the RCS at the $n$th and $\left(n-1\right)$th time epoches, respectively. Recalling the constant RCS assumption, it follows from (\ref{eq20}) and (\ref{eq26}) that
\begin{equation}\label{eq27}
\begin{gathered}
  {\beta _n} = {\beta _{n - 1}} \cdot \frac{{{\varepsilon _n}{d_{n - 1}}}}{{{\varepsilon _{n - 1}}{d_n}}} = {\beta _{n - 1}}\frac{{{d_{n - 1}}}}{{{d_n}}} \hfill \\
   \approx {\beta _{n - 1}}\left( {1 + \frac{{\Delta d\cos {\theta _{n - 1}}}}{{{d_n}}}} \right) \approx {\beta _{n - 1}}\left( {1 + \frac{{\Delta d\cos {\theta _{n - 1}}}}{{{d_{n - 1}}}}} \right). \hfill \\
\end{gathered}
\end{equation}
Finally, we summarize the state evolution model as
\begin{equation}\label{eq29}
\left\{ \begin{gathered}
  {\theta _n} = {\theta _{n - 1}} + d_{n - 1}^{ - 1}{v_{n - 1}}\Delta T\sin {\theta _{n - 1}} + {\omega _\theta }, \hfill \\
  {d_n} = {d_{n - 1}} - {v_{n - 1}}\Delta T\cos {\theta _{n - 1}} + {\omega _d}, \hfill \\
  {v_n} = {v_{n - 1}} + {\omega _v}, \hfill \\
  {\beta _n} = {\beta _{n - 1}}\left( {1 + d_{n - 1}^{ - 1}{v_{n - 1}}\Delta T\cos {\theta _{n - 1}}} \right) + {\omega _\beta }, \hfill \\
\end{gathered}  \right.
\end{equation}
where $\omega_\theta$, $\omega_d$, $\omega_v$ and ${\omega_{\beta} }$ denote the corresponding noises, which are assumed to be zero-mean Gaussian distributed with variances of ${\sigma _\theta^2,\sigma _d^2,\sigma _v^2}$ and $\sigma_{\beta}^2$, respectively. We highlight here that these noises are generated by approximation and other systematic errors, which are irrelevant to the measurement SNR defined in the last subsection.
\\\indent \emph{Remark 2:} To verify the performance of the approximations of $d_n$ and $\theta_n$, we plot in Fig. 3 both the real and the approximated values of the distance and the angle within 20 time slots, where the vehicle's speed is $54\text{km/h}$, and the length of each slot is $\Delta T = 100\text{ms}$. The initial values of the distance and the angle is 40m and $18^\circ$, respectively. One can see that even with a relatively large $\Delta d = 1.5\text{m}$, the approximation errors are still negligible in general. Therefore the associated variances could be very small.
\begin{figure}[!t]
    \centering
    \includegraphics[width=0.7\columnwidth]{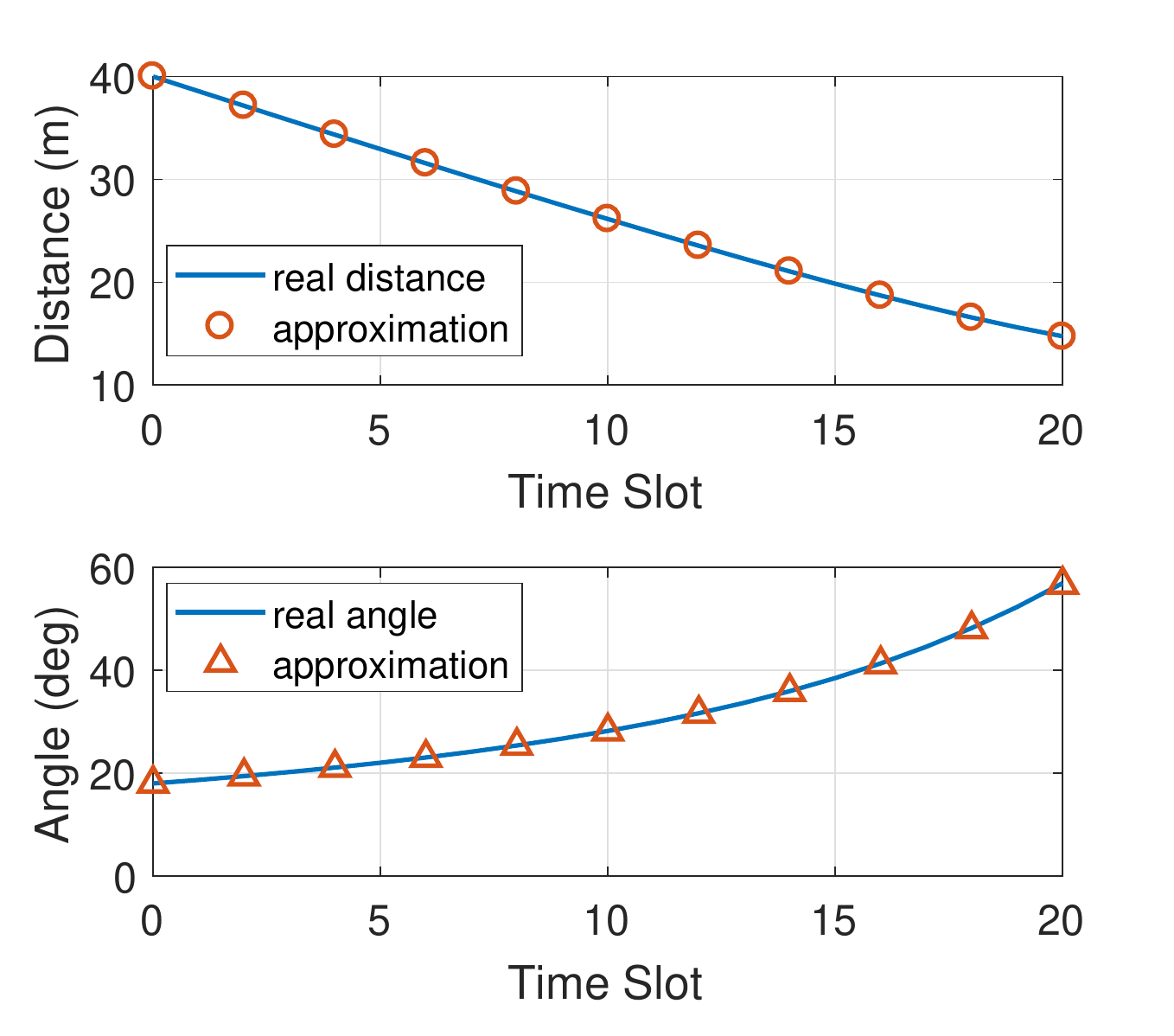}
    \caption{Verification of the approximation in (\ref{eq20}) and (\ref{eq24}).}
    \label{fig:3}
\end{figure}
\section{The Proposed Approach}
\subsection{Extended Kalman Filtering}
In this subsection, we propose a Kalman filtering scheme for beam prediction and tracking. Due to the nonlinearity in the measurement and the state evolution models, the linear Kalman filtering (LKF) can not be directly applied. We therefore consider an EKF approach that performs local linearization for nonlinear models. By denoting the state variables as ${\mathbf{x}} = {\left[ {\theta ,d,v,\beta} \right]^T}$ and the measured signal vector as ${\mathbf{y}} = {\left[ {{{{\mathbf{\tilde r}}}^T},\tau ,\mu } \right]^T}$, the models developed in (\ref{eq29}) and (\ref{eq8})-(\ref{eq10}) can be recast in compact forms as
\begin{equation}\label{eq30}
\left\{ \begin{gathered}
  {\text{State}}\;{\text{Evolution}}\;{\text{Model: }}{{\mathbf{x}}_n} = {\mathbf{g}}\left( {{{\mathbf{x}}_{n - 1}}} \right) + {{\bm{\omega }}_n}, \hfill \\
  {\text{Measurement}}\;{\text{Model:}}\;{{\mathbf{y}}_n} = {\mathbf{h}}\left( {{{\mathbf{x}}_n}} \right) + {{\mathbf{z}}_n}, \hfill \\
\end{gathered}  \right.
\end{equation}
where $\mathbf{g}\left(\cdot\right)$ is defined in (\ref{eq29}), with $\bm{\omega} = {\left[ {{\omega _\theta },{\omega _d},{\omega _v},{\omega _\beta }} \right]^T}$ being the noise vector that is independent to ${\mathbf{g}}\left( {{{\mathbf{x}}_{n - 1}}} \right)$. Similarly, $\mathbf{h}\left(\cdot\right)$ is defined as (\ref{eq8})-(\ref{eq10}), with ${\mathbf{z}} = {\left[ {{\mathbf{z}}_\theta ^T,{z_\tau },{z_f}} \right]^T}$ being the measurement noise that is independent to ${\mathbf{h}}\left( {{{\mathbf{x}}_n}} \right)$. As considered above, both $\bm{\omega}$ and $\mathbf{z}$ are zero-mean Gaussian distributed, with covariance matrices being expressed as
\begin{equation}\label{eq31}
  {{\mathbf{Q}}_s} = \operatorname{diag} \left( {\sigma _\theta^2,\sigma _d^2,\sigma _v^2,\sigma _{\beta}^2}\right),
\end{equation}
\begin{equation}\label{eq32}
  {{\mathbf{Q}}_m} = \operatorname{diag} \left( {\sigma _1^2{{\mathbf{1}}_{{N_r}}^T},\sigma _2^2,\sigma _3^2}\right),
\end{equation}
where $\mathbf{1}_{N_r}$ denotes a size-$N_r$ all one column vector. In order to linearize the models, the Jacobian matrices for both $\mathbf{g}\left(\mathbf{x}\right)$ and $\mathbf{h}\left(\mathbf{x}\right)$ need to be computed. By simple algebraic manipulation, the Jacobian matrix for $\mathbf{g}\left(\mathbf{x}\right)$ can be straightforwardly given as
\begin{equation}\label{eq33}
\begin{gathered}
  \frac{{\partial {\mathbf{g}}}}{{\partial {\mathbf{x}}}} = \hfill \\
  \left[ {\begin{array}{*{20}{c}}
  {1 + \frac{{v\Delta T\cos \theta }}{d}}&{ - \frac{{v\Delta T\sin \theta }}{{{d^2}}}}&{\frac{{\Delta T\sin \theta }}{d}}&0\\
  \scriptstyle{v\Delta T\sin \theta }&1&\scriptstyle{ - \Delta T\cos \theta }&0 \\
  0&0&1&0 \\
    - \frac{\beta{v\Delta T\sin \theta }}{d}&-\frac{{\beta v\Delta T\cos \theta }}{{{d^2}}}&\frac{{\beta \Delta T\cos \theta }}{d}&1 + \frac{{v\Delta T\cos \theta }}{d}
\end{array}} \right].
\end{gathered}
\end{equation}
For $\mathbf{h}\left(x\right)$, let us denote
\begin{equation}\label{eq34}
  {\bm{\eta }}\left( {\beta ,\theta } \right) = \kappa\beta {\mathbf{b}}\left( \theta  \right){{\mathbf{a}}^H}\left( \theta  \right){\mathbf{a}}\left( {\hat \theta } \right),
\end{equation}
where $\hat \theta$ is a prediction for $\theta$. The Jacobian matrix for $\mathbf{h}\left(x\right)$ can be then given by
\begin{equation}\label{eq34}
\frac{{\partial {\mathbf{h}}}}{{\partial {\mathbf{x}}}} = \left[ {\begin{array}{*{20}{c}}
  {\frac{{\partial {\bm\eta} }}{{\partial \theta }}}&0&0&{\frac{{\partial {\bm\eta} }}{{\partial {\beta }}}} \\
  0&{\frac{2}{c}}&0&0 \\
  { - \frac{{2v\sin \theta }}{c}}&0&{\frac{{2{f_c}\cos \theta }}{c}}&0
\end{array}} \right].
\end{equation}
It remains to compute ${\frac{{\partial {\bm{\eta }}}}{{\partial \theta }}}$ and ${\frac{{\partial {\bm{\eta }}}}{{\partial {\beta} }}}$. First of all, note that the derivatives with respect to $\beta$ can be simply calculated as
\begin{equation}\label{eq35}
\frac{{\partial \bm{\eta } }}{{\partial {\beta}}} = \kappa{\mathbf{b}}\left( \theta  \right){{\mathbf{a}}^H}\left( \theta  \right){\mathbf{a}}\left( {\hat \theta } \right).
\end{equation}
We then expand ${\bm{\eta }}\left( {\beta ,\theta } \right)$ as
\begin{equation}\label{eq36}
  {\bm{\eta }} = \frac{\beta}{\sqrt{N_t}}\left[ \begin{gathered}
  \sum\limits_{i = 1}^{{N_t}} {{e^{-j\pi \left( {i - 1} \right)\cos \hat \theta }}{e^{ j\pi \left( {i - 1} \right)\cos \theta }}}  \hfill \\
  \sum\limits_{i = 1}^{{N_t}} {{e^{-j\pi \left( {i - 1} \right)\cos \hat \theta }}{e^{ j\pi \left( {i - 2} \right)\cos \theta }}}  \hfill \\
   \vdots  \hfill \\
  \sum\limits_{i = 1}^{{N_t}} {{e^{-j\pi \left( {i - 1} \right)\cos \hat \theta }}{e^{ j\pi \left( {i - {N_r}} \right)\cos \theta }}}  \hfill \\
\end{gathered}  \right].
\end{equation}
It follows that
\begin{equation}\label{eq37}
\begin{gathered}
  \frac{{\partial {\bm{\eta }}}}{{\partial \theta }} =  \hfill \\
  \frac{\beta}{\sqrt{N_t}}\left[ \begin{gathered}
   - \sum\limits_{i = 1}^{{N_t}} {{e^{-j\pi \left( {\left( {i - 1} \right)\cos \hat \theta  - \left( {i - 1} \right)\cos\theta } \right)}}j\pi \left( {i - 1} \right)\sin \theta }  \hfill \\
   - \sum\limits_{i = 1}^{{N_t}} {{e^{-j\pi \left( {\left( {i - 1} \right)\cos \hat \theta  - \left( {i - 2} \right)\cos \theta } \right)}}j\pi \left( {i - 2} \right)\sin \theta }  \hfill \\
   \vdots  \hfill \\
   - \sum\limits_{i = 1}^{{N_t}} {{e^{-j\pi \left( {\left( {i - 1} \right)\cos \hat \theta  - \left( {i - {N_r}} \right)\cos \theta } \right)}}j\pi \left( {i - {N_r}} \right)\sin \theta }  \hfill \\
\end{gathered}  \right]. \hfill \\
\end{gathered}
\end{equation}
We are now ready to present the EKF technique. Following the standard procedure of Kalman filtering \cite{kay1998fundamentals}, the state prediction and tracking design is summarized as follows:
\begin{enumerate}
  \item \emph{State Prediction:}
  \begin{equation}\label{eq38}
    {{{\mathbf{\hat x}}}_{n\left| {n - 1} \right.}} = {\mathbf{g}}\left( {{{{\mathbf{\hat x}}}_{n - 1}}} \right),{{{\mathbf{\hat x}}}_{n + 1\left| {n - 1} \right.}} = {\mathbf{g}}\left( {{{{\mathbf{\hat x}}}_{n\left| {n - 1} \right.}}} \right).
  \end{equation}
  \item \emph{Linearization:}
  \begin{equation}\label{eq39}
    {{\mathbf{G}}_{n - 1}} = {\left. {\frac{{\partial {\mathbf{g}}}}{{\partial {\mathbf{x}}}}} \right|_{{\mathbf{x}} = {{{\mathbf{\hat x}}}_{n - 1}}}},{{\mathbf{H}}_n} = {\left. {\frac{{\partial {\mathbf{h}}}}{{\partial {\mathbf{x}}}}} \right|_{{\mathbf{x}} = {{{\mathbf{\hat x}}}_{n\left| {n - 1} \right.}}}}.
  \end{equation}
  \item \emph{MSE Matrix Prediction:}
  \begin{equation}\label{eq40}
    {{\mathbf{M}}_{n\left| {n - 1} \right.}} = {\mathbf{G}}_{n - 1}{{\mathbf{M}}_{n - 1}}{\mathbf{G}}_{n - 1}^H + {{\mathbf{Q}}_s}.
  \end{equation}
  \item \emph{Kalman Gain Calculation:}
  \begin{equation}\label{eq41}
    {{\mathbf{K}}_n} = {{\mathbf{M}}_{n\left| {n - 1} \right.}}{\mathbf{H}}_n^H{\left( {{{\mathbf{Q}}_m} + {{\mathbf{H}}_n}{{\mathbf{M}}_{n\left| {n - 1} \right.}}{\mathbf{H}}_n^H} \right)^{ - 1}}.
  \end{equation}
  \item \emph{State Tracking:}
  \begin{equation}\label{eq42}
    {{{\mathbf{\hat x}}}_n} = {{{\mathbf{\hat x}}}_{n\left| {n - 1} \right.}} + {{\mathbf{K}}_n}\left( {{{\mathbf{y}}_n} - {\mathbf{h}}\left( {{{{\mathbf{\hat x}}}_{n\left| {n - 1} \right.}}} \right)} \right).
  \end{equation}
  \item \emph{MSE Matrix Update:}
  \begin{equation}\label{eq43}
    {{\mathbf{M}}_n} = \left( {{\mathbf{I}} - {{\mathbf{K}}_n}{{\mathbf{H}}_n}} \right){{\mathbf{M}}_{n\left| {n - 1} \right.}}.
  \end{equation}
\end{enumerate}

\emph{Remark 3:} In the prediction step, the predicted angle ${\hat \theta}_{n\left| {n - 1} \right.}$ is used for transmit beamforming at the RSU at the $n$th epoch, ${\hat \theta}_{n+1\left| {n - 1} \right.}$ is sent to the vehicle for receive beamforming at the $\left(n+1\right)$th epoch. In the tracking step, based on the received target echo $\mathbf{y}_n$, the RSU refines the predicted state $\mathbf{\hat x}_{n\left| {n - 1} \right.}$ to obtain $\mathbf{\hat x}_n$, which will be used as the input of the predictor for the next iteration. By iteratively performing prediction and tracking, the RSU is able to simultaneously sense and communicate with the vehicle.
\subsection{Beam Association for Multiple Vehicles}
After processing the echoes received from $K$ beams individually, the RSU obtains the angle estimates for $K$ vehicles, which we denote as ${{{\bm{\hat \theta }}}_n} = \left[ {{{\hat \theta }_{1,n}},{{\hat \theta }_{2,n}}, \ldots {{\hat \theta }_{K,n}}} \right]$. To transmit the desired information at each beam, the RSU needs to correctly map the vehicles to the corresponding beams. In a communication-only beam tracking scenario, this can be realized by simply letting each vehicle send its own ID information to the RSU along with the uplink beam training signal at each epoch. However, such association can not be done in a straightforward manner in the joint radar-communication scenario considered, since the reflected echoes contain no ID information of the vehicles. Even if the vehicles are able to transmit the IDs through uplink communication, such information might be out-dated due to the high-dynamic vehicular channel. To tackle this issue, we propose in the following a simple and efficient beam association approach, which tracks the ID of each vehicle based on the corresponding state parameters, i.e., $\theta$, $d$, $v$ and $\beta$.
\\\indent Assume that at the initial stage, i.e., at the $0$th epoch, the RSU is able to obtain the initial ID information and the estimated angles ${{{\bm{\hat \theta }}}_0}$ from the initial uplink training. Our goal is then to associate the angles in ${{{\bm{\hat \theta }}}_n}$ with that of ${{{\bm{\hat \theta }}}_{n-1}}$ during each Kalman iteration. Let us assume that ${\hat \theta}_{i,n}$ is associated with ${\hat \theta}_{j,n-1}$. By taking a closer look at the scenario considered, we note that the corresponding state variation, i.e., $\left\| {{{\mathbf{x}}_{i,n}} - {{\mathbf{x}}_{j,n - 1}}} \right\|$, will be very small. As an example, we see from (\ref{eq20}) that the distance traveled within the time slot $\Delta T$ by a single vehicle can be approximated by
\begin{equation}\label{eq44}
  \left| {\Delta d} \right| \approx v\Delta T\cos \theta  \le v\Delta T.
\end{equation}
For a moderate driving speed $v = 15\text{m/s}$ and $\Delta T = 10\text{ms}$, we have $\left| {\Delta d_{i,j}} \right| \le 0.15\text{m}$. Nevertheless, it is not possible for two vehicles having such a small distance given their size as well as the safety distance that needs to be maintained. The above observation can be mathematically expressed as
\begin{equation}\label{eq55}
\left\| {{{\mathbf{x}}_{i,n}} - {{\mathbf{x}}_{j,n - 1}}} \right\| \ll \left\| {{{\mathbf{x}}_{i,n}} - {{\mathbf{x}}_{k,n - 1}}} \right\|,\forall k \ne j.
\end{equation}
Therefore, the beam association can be done by simply finding ${\mathbf x}_{j,n-1}$ that has the minimum Euclidean distance to ${\mathbf x}_{i,n}$, and then associate ${\hat \theta}_{i,n}$ with ${\hat \theta}_{j,n-1}$ accordingly.
\section{Multi-Beam Power Allocation for Joint Sensing and Communication}
At each epoch, the RSU will send joint radar-communication signals to $K$ vehicles via $K$ beams, whose reflections will be exploited for beam tracking. Furthermore, to maintain reliable communication links, the downlink quality-of-service (QoS) should be guaranteed. In this section, we propose a multi-beam power allocation scheme for improving the sensing performance for the vehicles while guaranteeing the sum-rate of V2I downlink communications. In particular, we are interested in minimizing the estimation errors with respect to the angle $\theta$ and the distance $d$.
\subsection{Water-filling Power Allocation: A Benchmark Method}
Before presenting our power allocation scheme, let us firstly review the classic water-filling approach, which will serve as the benchmark technique in the multi-vehicle scenario considered. To be specific, the water-filling solution is obtained by maximizing the sum-rate of all the $K$ vehicles under a given power budget, which requires to solve the following optimization problem:
\begin{equation}\label{eq_wf}
\begin{gathered}
  \mathop {\max }\limits_{{{\mathbf{p}}_n}} \;{R_n} = \sum\limits_{k = 1}^K {{{\log }_2}\left( {1 + {p_{k,n}}{\rho _{k,n}}} \right)}  \hfill \\
  s.t.\;\;\;{{\mathbf{1}}^T}{{\mathbf{p}}_n} \le {P_T},{p_{k,n}} \ge 0,\forall k, \hfill \\
\end{gathered}
\end{equation}
where ${{\mathbf{p}}_n} = {\left[ {{p_{1,n}},{p_{2,n}},...,{p_{K,n}}} \right]^T}$ is the power allocation vector, $R_n$ is the sum-rate defined by (\ref{eq15}) with $\rho_{k,n}$ being given by (\ref{eq14}), and finally $P_T$ is the transmit power budget of the RSU. It is well-known that the optimal solution for (\ref{eq_wf}) is given by \cite{boyd2004convex}
\begin{equation}\label{eq_wf_solution}
{p_{k,n}} = {\left( {\gamma  - \rho _{k,n}^{ - 1}} \right)^ + } = \max \left\{ {0,\gamma  - \rho _{k,n}^{ - 1}} \right\},
\end{equation}
where $\gamma$ is chosen such that the power constraint is satisfied. While the water-filling method is able to maximize the communication rate, we note here that it does not address the issue of minimizing the estimation errors in the V2I link, which we discuss in the following.
\subsection{Posterior Cram\'er-Rao Bound for Parameter Estimation}
To characterize the estimation performance, one needs to derive the Cram\'er-Rao bound (CRB) of parameter estimation, which serves as a lower-bound for the variances of any unbiased estimators \cite{kay1998fundamentals}. Unlike conventional CRB that relies on the measured data only, the CRB for vehicle parameter estimation needs to consider both the measurement and the state models, leading to a posterior Cram\'er-Rao bound (PCRB) \cite{5571900}. Given a measurement $\mathbf{y}_n$ with respect to a state $\mathbf{x}_n$, based on Bayes' theorem, the joint probability density function (PDF) of $\mathbf{x}_n$ and $\mathbf{y}_n$ can be expressed as \cite{kay1998fundamentals}
\begin{equation}\label{eq56}
p\left( {{{\mathbf{x}}_n},{{\mathbf{y}}_n}} \right) = p\left( {{{\mathbf{y}}_n}\left| {{{\mathbf{x}}_n}} \right.} \right)p\left( {{{\mathbf{x}}_n}} \right),
\end{equation}
where $p\left( {{{\mathbf{y}}_n}\left| {{{\mathbf{x}}_n}} \right.} \right)$ is the conditional PDF of $\mathbf{y}_n$ given $\mathbf{x}_n$, and $p\left( {{{\mathbf{x}}_n}} \right)$ is the prior PDF of $\mathbf{x}_n$, which is determined by the state evolution model. For the measurement $\mathbf{y}_n$, we have
\begin{equation}\label{eq57}
\begin{gathered}
  p\left( {{{\mathbf{y}}_n}\left| {{{\mathbf{x}}_n}} \right.} \right) = \hfill \\
    \frac{1}{{{\pi ^{{N_r} + 2}}\det \left( {{{\mathbf{Q}}_m}} \right)}}\exp \left( {{{\left( {{{\mathbf{y}}_n} - {\mathbf{h}}\left( {{{\mathbf{x}}_n}} \right)} \right)}^H}{\mathbf{Q}}_m^{ - 1}\left( {{{\mathbf{y}}_n} - {\mathbf{h}}\left( {{{\mathbf{x}}_n}} \right)} \right)} \right). \hfill \\
\end{gathered}
\end{equation}
On the other hand, $\mathbf{x}_n$ is dependent on $\mathbf{x}_{n-1}$, i.e., ${{\mathbf{x}}_n} = {\mathbf{g}}\left( {{{\mathbf{x}}_{n - 1}}} \right) + {{\bm\omega} _n}$, given ${{\mathbf{x}}_{n - 1}} \sim \mathcal{CN}\left( {{{{\mathbf{\hat x}}}_{n - 1}},{{\mathbf{M}}_{n - 1}}} \right)$. Due to the nonlinearity in $\mathbf{g}$, it is difficult to analytically derive the distribution of $\mathbf{x}_n$. We therefore resort to a linear approximation following the linearization step in the EKF procedure, in which case we have
\begin{equation}
  {{\mathbf{x}}_n} \approx {{\mathbf{G}}_{n - 1}}{{\mathbf{x}}_{n - 1}} + {{\bm\omega}_n},
\end{equation}
where ${{\mathbf{G}}_{n - 1}}$ is the Jacobian matrix defined in (\ref{eq39}). Since both $\mathbf{x}_{n-1}$ and ${{\bm\omega}_n}$ are Gaussian distributed and are independent with each other, ${{\mathbf{x}}_n}$ is Gaussian distributed as well, which subjects to
\begin{equation}
  {{\mathbf{x}}_n} \sim \mathcal{CN}\left( {{{\bm{\mu }}_n},{{\mathbf{M}}_{\left. n \right|n - 1}}} \right),{{\bm{\mu }}_n} = {{\mathbf{G}}_{n - 1}}{{{\mathbf{\hat x}}}_{n - 1}},
\end{equation}
where ${{\mathbf{M}}_{\left. n \right|n - 1}}$ is defined in (\ref{eq40}). As a result, the prior PDF $p\left({{\mathbf{x}}_n}\right)$ can be expressed as
\begin{equation}
\begin{gathered}
  p\left( {{{\mathbf{x}}_n}} \right) =  \hfill \\
  \frac{1}{{{\pi ^4}\det \left( {{{\mathbf{M}}_{\left. n \right|n - 1}}} \right)}}\exp \left( {{{\left( {{{\mathbf{x}}_n} - {{\bm{\mu }}_n}} \right)}^H}{\mathbf{M}}_{\left. n \right|n - 1}^{ - 1}\left( {{{\mathbf{x}}_n} - {{\bm{\mu }}_n}} \right)} \right). \hfill \\
\end{gathered}
\end{equation}
According to \cite{5571900}, the posterior Fisher information matrix (FIM) of $\mathbf{x}_n$ can be obtained by
\begin{equation}\label{eq57}
\begin{gathered}
  {\mathbf{J}} = - \mathbb{E}\left( {\frac{{{\partial ^2}\ln p\left( {{{\mathbf{x}}_n},{{\mathbf{y}}_n}} \right)}}{{\partial {\mathbf{x}}_n^2}}} \right) = {{\mathbf{J}}_{m}} + {{\mathbf{J}}_{s}} \hfill \\
   = \underbrace { - \mathbb{E}\left( {\frac{{{\partial ^2}\ln p\left( {{\mathbf{y}}_n\left| {\mathbf{x}_n} \right.} \right)}}{{\partial {{\mathbf{x}}_n^2}}}} \right)}_{{\text{Observed}}\;{\text{Fisher}}\;{\text{Information}}}\underbrace { - \mathbb{E}\left( {\frac{{{\partial ^2}\ln p\left( {\mathbf{x}_n} \right)}}{{\partial {{\mathbf{x}}_n^2}}}} \right)}_{{\text{Prior}}\;{\text{Fisher}}\;{\text{Information}}} \hfill \\
\end{gathered} ,
\end{equation}
where ${{\mathbf{J}}_{m}}$ and ${{\mathbf{J}}_{s}}$ are FIMs with respect to $p\left( {{{\mathbf{y}_n}}\left| {{{\mathbf{x}_n}}} \right.} \right)$ and $p\left( {{{\mathbf{x}}_n}} \right)$. Eq. (\ref{eq57}) implies that both the measurement $\mathbf{y}_n$ and the \emph{a priori} knowledge, i.e., the state evolution model, will provide Fisher information of $\mathbf{x}_n$. For the Gaussian PDFs above, the FIMs can be readily computed as
\begin{equation}\label{eq58}
\begin{gathered}
  {{\mathbf{J}}_{m}} = {\left( {\frac{{\partial {\mathbf{h}}}}{{\partial {{\mathbf{x}}_n}}}} \right)^H}{\mathbf{Q}}_m^{ - 1}\left( {\frac{{\partial {\mathbf{h}}}}{{\partial {{\mathbf{x}}_n}}}} \right), \hfill \\
  {{\mathbf{J}}_{s}} = {\mathbf{M}}_{\left. n \right|n - 1}^{ - 1} = {\left( {{{\mathbf{G}}_{n - 1}}{{\mathbf{M}}_{n - 1}}{\mathbf{G}}_{n - 1}^H + {{\mathbf{Q}}_s}} \right)^{ - 1}}. \hfill \\
\end{gathered}
\end{equation}
\\\indent \emph{Remark 4:} Note that ${{\mathbf{M}}_{\left. n \right|n - 1}}$ in (\ref{eq40}) consists of the state covariance matrix $\mathbf{Q}_s$ as well as the covariance matrix $\mathbf{M}_{n-1}$ of the $\left(n-1\right)$th estimation. This suggests that an accurate state model and an accurate estimate of $\mathbf{x}_{n-1}$ (which lead to ``small" $\mathbf{Q}_s$ and $\mathbf{M}_{n-1}$) will provide a large amount of Fisher information of $\mathbf{x}_n$.
\\\indent Based on the definition of CRB, the MSE matrix of $\mathbf{x}$ is lower bounded by the inverse of $\mathbf{J}$, which is
\begin{equation}\label{eq59}
\mathbb{E}\left( {\left( {{\mathbf{\hat x}_n} - {\mathbf{x}_n}} \right){{\left( {{\mathbf{\hat x}_n} - {\mathbf{x}_n}} \right)}^H}} \right) \succeq {{\mathbf{J}}^{ - 1}} \triangleq \mathbf{C}.
\end{equation}
Accordingly, the MSEs of angle and distance are bounded by
\begin{equation}\label{eq60}
\begin{gathered}
  \mathbb{E}\left( {{{\left( {{{\hat \theta }_n} - {\theta _n}} \right)}^2}} \right) \ge {c_{11}} \triangleq \operatorname{PCRB} \left( {{\theta _n}} \right), \hfill \\
  \mathbb{E}\left( {{{\left( {{{\hat d}_n} - {d_n}} \right)}^2}} \right) \ge {c_{22}} \triangleq \operatorname{PCRB} \left( {{d_n}} \right), \hfill \\
\end{gathered}
\end{equation}
where $c_{ij}$ denotes the $\left(i,j\right)$th entry of $\mathbf{C}$.

\subsection{Problem Formulation and Analysis}
In practice, since the real values of $\theta$ and $d$ are always unknown, one can only resort to predicted parameters for PCRB optimization, resulting in predicted FIM and PCRB. Given the predictions ${{{\hat \theta }_{k,\left. n \right|n - 1}}}$ and ${{{\hat d }_{k,\left. n \right|n - 1}}}$ for the $k$th vehicle at the $n$th epoch, the predicted PCRBs are obtained in the form
\begin{equation}\label{eq61}
\begin{gathered}
  \operatorname{PCRB} \left( {{{\hat \theta }_{k,\left. n \right|n - 1}}} \right) = {\left. {{c_{11}}} \right|_{{\mathbf{x}} = {{{\mathbf{\hat x}}}_{k,n\left| {n - 1} \right.}}}}, \hfill \\
  \operatorname{PCRB} \left( {{{\hat d}_{k,\left. n \right|n - 1}}} \right) = {\left. {{c_{22}}} \right|_{{\mathbf{x}} = {{{\mathbf{\hat x}}}_{k,n\left| {n - 1} \right.}}}}, \hfill \\
\end{gathered}
\end{equation}
where ${{{c}_{11}}}$ and ${{{c}_{22}}}$ are calculated by substituting ${{\mathbf{x}} = {{{\mathbf{\hat x}}}_{k,n\left| {n - 1} \right.}}}$ into (\ref{eq57})-(\ref{eq59}).
\begin{thm}
The predicted PCRB matrix equals to the updated MSE matrix in the Kalman iteration, i.e.,
  \begin{equation}
    {\mathbf{C}}\left| {_{{{\mathbf{x}}_n} = {{{\mathbf{\hat x}}}_{n\left| {n - 1} \right.}}}} \right. = {{\mathbf{M}}_n}.
  \end{equation}
\end{thm}
\renewcommand{\qedsymbol}{$\blacksquare$}
\begin{proof}
For EKF, the following identity holds true according to \cite{1101979}
\begin{equation}\label{identity_EKF}
  {\mathbf{M}}_n^{ - 1} = {\left( {{{\mathbf{G}}_{n - 1}}{{\mathbf{M}}_{n - 1}}{\mathbf{G}}_{n - 1}^H + {{\mathbf{Q}}_s}} \right)^{ - 1}} + {\mathbf{H}}_n^H{\mathbf{Q}}_m^{ - 1}{{\mathbf{H}}_n}.
\end{equation}
Comparing (\ref{identity_EKF}) with (\ref{eq57}) and (\ref{eq58}), it can be immediately observed that substituting ${{{\mathbf{\hat x}}}_{n\left| {n - 1} \right.}}$ to the Jacobian matrix ${\frac{{\partial {\mathbf{h}}}}{{\partial {{\mathbf{x}}_n}}}}$ yields
\begin{equation}
  {\mathbf{J}}\left| {_{{{\mathbf{x}}_n} = {{{\mathbf{\hat x}}}_{n\left| {n - 1} \right.}}}} \right. = {\mathbf{M}}_n^{ - 1}.
\end{equation}
This completes the proof.
\end{proof}
\indent \emph{Remark 5:} Based on Theorem 1, the predicted PCRBs (\ref{eq61}) are equivalent to the first and the second entries in $\mathbf{M}_n$. Nevertheless, since both the MSE matrix $\mathbf{M}_n$ and the PCRB are approximations of their real counterparts due to the linearization steps involved, (\ref{eq61}) should still be regarded as approximated bounds of the MSEs rather than their real values.
\\\indent Since the estimations of each $\theta_{k,n}$ and $d_{k,n}$ are independent with each other, the joint PCRB for $K$ vehicles can be given by the following summation as
\begin{equation}\label{eq62}
\begin{gathered}
  \sum\limits_{k = 1}^K {\operatorname{PCRB} \left( {{{\hat \theta }_{k,\left. n \right|n - 1}}} \right)}  = \sum\limits_{k = 1}^K {{{\left. {{c_{11}}} \right|}_{{\mathbf{x}} = {{{\mathbf{\hat x}}}_{k,n\left| {n - 1} \right.}}}}}, \hfill \\
  \sum\limits_{k = 1}^K {\operatorname{PCRB} \left( {{{\hat d}_{k,\left. n \right|n - 1}}} \right)}  = \sum\limits_{k = 1}^K {{{\left. {{c_{22}}} \right|}_{{\mathbf{x}} = {{{\mathbf{\hat x}}}_{k,n\left| {n - 1} \right.}}}}}. \hfill \\
\end{gathered}
\end{equation}
Our goal is to optimally allocate transmit power among multiple beams, such that the joint PCRB (\ref{eq62}) can be minimized while ensuring the downlink sum-rate of the communication. This problem can be formulated as
\begin{equation}\label{eq63}
\begin{gathered}
  \mathop {\min }\limits_{{{\mathbf{p}}_n}} \;\sum\limits_{k = 1}^K {\left( {\operatorname{PCRB} \left( {{{\hat \theta }_{k,\left. n \right|n - 1}}} \right) + \operatorname{PCRB} \left( {{{\hat d}_{k,\left. n \right|n - 1}}} \right)} \right)}   \hfill \\
  s.t.\;\;{R_n} \ge R_t,{{\mathbf{1}}^T}{{\mathbf{p}}_n} \le {P_T}, p_{k,n} \ge 0, \forall k, \hfill \\
\end{gathered}
\end{equation}
where $R_t$ is a required sum-rate threshold, and all the other notations are defined as the same in (\ref{eq_wf}). While the dependency of $R_n$ and $p_{k,n},\forall k$ has been explicitly shown in (\ref{eq15}), it remains to reveal the relationship between ${\left. {{{c}_{11}}} \right|_{{\mathbf{x}} = {{{\mathbf{\hat x}}}_{k,n\left| {n - 1} \right.}}}}, {\left. {{{c}_{22}}} \right|_{{\mathbf{x}} = {{{\mathbf{\hat x}}}_{k,n\left| {n - 1} \right.}}}}$ and $p_{k,n}$. Let us denote \footnote{Note that since the real values of $\beta _{k,n}$ and $\delta _{k,n}$ are unknown to the RSU, one has to resort to predicted values, i.e., ${\beta _{k,n}} = {{\hat \beta }_{k,n\left| {n - 1} \right.}},{\delta _{k,n}} = {{\hat \delta }_{k,n\left| {n - 1} \right.}} = 1$, to estimate ${{{\mathbf{\tilde Q}}}_m}$ in practice.}
\begin{equation}\label{eq64}
\begin{gathered}
  {{\mathbf{Q}}_m} =
  p_{k,n}^{ - 1}\operatorname{diag} \left( {\frac{{{a_1^2}{\sigma ^2}}}{G},\frac{{{a_2^2}{\sigma ^2}}}{{G{\kappa ^2}{{\left| {{\beta _{k,n}}{\delta _{k,n}}} \right|}^2}}},\frac{{{a_3^2}{\sigma ^2}}}{{G{\kappa ^2}{{\left| {{\beta _{k,n}}{\delta _{k,n}}} \right|}^2}}}} \right) \hfill \\
   \triangleq p_{k,n}^{ - 1}{{{\mathbf{\tilde Q}}}_m}. \hfill \\
\end{gathered}
\end{equation}
It then follows that
\begin{equation}\label{eq65}
\begin{gathered}
  {\mathbf{J}} = {{\mathbf{J}}_m}\left| {_{{{\mathbf{x}}_n} = {{{\mathbf{\hat x}}}_{n\left| {n - 1} \right.}}} + } \right.{{\mathbf{J}}_s} \hfill \\
   = {p_{k,n}}{\mathbf{H}}_n^H{\mathbf{\tilde Q}}_m^{ - 1}{{\mathbf{H}}_n} + {\mathbf{M}}_{\left. n \right|n - 1}^{ - 1} \triangleq {p_{k,n}}{\mathbf{A}} + {\mathbf{B}}, \hfill \\
\end{gathered}
\end{equation}
where $\mathbf{H}_n$ is defined in (\ref{eq39}), ${\mathbf{A}} = {\mathbf{H}}_n^H{\mathbf{\tilde Q}}_m^{ - 1}{{\mathbf{H}}_n}$ and ${\mathbf{B}} = {\mathbf{M}}_{\left. n \right|n - 1}^{ - 1}$. The inverse of the FIM can be given by
\begin{equation}\label{eq66}
\begin{gathered}
  {\mathbf{C}} = {\left( {{p_{k,n}}{\mathbf{A}} + {\mathbf{B}}} \right)^{ - 1}} = {\left( {{p_{k,n}}{\mathbf{A}} + {{\mathbf{B}}^{1/2}}{{\mathbf{B}}^{H/2}}} \right)^{ - 1}} \hfill \\
   = {{\mathbf{B}}^{ - H/2}}{\left( {{p_{k,n}}{{\mathbf{B}}^{ - 1/2}}{\mathbf{A}}{{\mathbf{B}}^{ - H/2}} + {\mathbf{I}}} \right)^{ - 1}}{{\mathbf{B}}^{ - 1/2}}, \hfill \\
\end{gathered}
\end{equation}
where ${\mathbf{B}}^{1/2}$ is a $4\times 4$ square-root of $\mathbf{B}$, with ${\mathbf{B}}^{H/2}$ being its Hermitian transpose. Both ${\mathbf{B}}^{1/2}$ and its Hermitian transpose are full-rank matrices and are thus invertible, since
\begin{equation}\label{eq67}
\begin{gathered}
  4 \ge \operatorname{rank} \left( {{{\mathbf{B}}^{1/2}}} \right) = \operatorname{rank} \left( {{{\mathbf{B}}^{H/2}}} \right) \hfill \\
   \ge \operatorname{rank} \left( {{{\mathbf{B}}^{1/2}}{{\mathbf{B}}^{H/2}}} \right) = \operatorname{rank} \left( {\mathbf{B}} \right) = 4. \hfill \\
\end{gathered}
\end{equation}
By taking the eigenvalue decomposition of ${{{\mathbf{B}}^{ - 1/2}}{\mathbf{A}}{{\mathbf{B}}^{ - H/2}}}$ we have
\begin{equation}\label{eq67}
  {\mathbf{U}}{\bm \Lambda} {{\mathbf{U}}^H} = {{\mathbf{B}}^{ - 1/2}}{\mathbf{A}}{{\mathbf{B}}^{ - H/2}},
\end{equation}
where $\mathbf{U}$ is an orthogonal matrix that contains the eigenvectors, ${\bm \Lambda}$ is a diagonal matrix composed by eigenvalues. Substituting (\ref{eq67}) into (\ref{eq66}) yields
\begin{equation}\label{eq68}
\begin{gathered}
  {\mathbf{C}} = {{\mathbf{B}}^{ - H/2}}{\left( {p_{k,n}{\mathbf{U}}{\bm \Lambda} {{\mathbf{U}}^H} + {\mathbf{I}}} \right)^{ - 1}}{{\mathbf{B}}^{ - 1/2}} \hfill \\
   = {{\mathbf{B}}^{ - H/2}}{\mathbf{U}}{\left( {p_{k,n}{\bm \Lambda}  + {\mathbf{I}}} \right)^{ - 1}}{{\mathbf{U}}^H}{{\mathbf{B}}^{ - 1/2}} \triangleq {\mathbf{\tilde B}}{\left( {p_{k,n}{\bm \Lambda}  + {\mathbf{I}}} \right)^{ - 1}}{{{\mathbf{\tilde B}}}^H}, \hfill \\
\end{gathered}
\end{equation}
where ${{\mathbf{\tilde B}}} = {{\mathbf{B}}^{ - H/2}}{\mathbf{U}}$. Given the diagonal structure of $\bm{\Lambda}$, we have
\begin{equation}\label{eq69}
  {\left( {p_{k,n}\bm{\Lambda}  + {\mathbf{I}}} \right)^{ - 1}} = \operatorname{diag} \left( {\frac{1}{{p_{k,n}{\lambda _{1,k,n}} + 1}}, \ldots \frac{1}{{p_{k,n}{\lambda _{4,k,n}} + 1}}} \right),
\end{equation}
where $\lambda_{i,k,n},i = 1,...,4$ are the eigenvalues in $\bm{\Lambda}$. The $\left(i,j\right)$th entry of $\mathbf{C}$ can be directly computed as
\begin{equation}\label{eq70-1}
  {c_{ij}} = \sum\limits_{m = 1}^4 {\frac{{{{\tilde b}_{im}}{{\tilde b}_{jm}^*}}}{{p{\lambda _{m,k,n}} + 1}}},
\end{equation}
where ${\tilde b}_{ij}$ is the $\left(i,j\right)$th entry of ${{\mathbf{\tilde B}}}$. Accordingly, the PCRBs for $\theta_{k,n}$ and $d_{k,n}$ are
\begin{equation}\label{eq70}
  {c_{11}} = \sum\limits_{m = 1}^4 {\frac{{{{\left| {{{\tilde b}_{1m}}} \right|}^2}}}{{{p_{k,n}}{\lambda _{m,k,n}} + 1}}}, {c_{22}} = \sum\limits_{m = 1}^4 {\frac{{{{\left| {{{\tilde b}_{2m}}} \right|}^2}}}{{{p_{k,n}}{\lambda _{m,k,n}} + 1}}}.
\end{equation}
Problem (\ref{eq63}) can be therefore recast as
\begin{equation}\label{eq71}
\begin{gathered}
  \mathop {\min }\limits_{{{\mathbf{p}}_n}} \sum\limits_{k = 1}^K {\sum\limits_{m = 1}^4 {\frac{\left|{\tilde b_{1m,k,n}}\right|^2+\left|{\tilde b_{2m,k,n}}\right|^2}{{{p_{k,n}}{\lambda _{m,k,n}} + 1}}} }  \hfill \\
  s.t.\;\;\sum\limits_{k = 1}^K {{{\log }_2}\left( {1 + {\rho_{k,n}}{p_{k,n}}} \right)}  \ge R_t, {{\mathbf{1}}^T}{{\mathbf{p}}_n} \le {P_T}, p_{k,n} \ge 0, \forall k, \hfill \\
\end{gathered}
\end{equation}
where $\tilde b_{1m,k,n}, \tilde b_{2m,k,n}$ and $\lambda _{m,k,n}$ are obtained by substituting the prediction ${{\mathbf{x}} = {{{\mathbf{\hat x}}}_{k,n\left| {n - 1} \right.}}}$ into (\ref{eq65})-(\ref{eq68}), and $\rho _{k,n}$ is defined in (\ref{eq14}). It can be readily observed that by solving (\ref{eq71}), the power budget $P_T$ will always be reached. This can be proved by contradiction. Assume that the summation of the optimal $p_{k,n}, \forall k$ is less than $P_T$. One can then increase any $p_{k,n}$ to reach a summation of $P_T$, while further reducing the objective function and increasing the sum-rate. Therefore, the optimal power allocation should fully exploit the budget $P_T$.
\\\indent Since we focus on the power allocation for each epoch, the time index in (\ref{eq71}) can be omitted, i.e., $\tilde b_{1m,k,n}$, $\tilde b_{2m,k,n}$, $\lambda _{m,k,n}$, $\rho _{k,n}$ and $p_{k,n}$ will be represented by $\tilde b_{1m,k}$, $\tilde b_{2m,k}$, $\lambda _{m,k}$, $\rho _{k}$ and $p_k$, respectively. Problem (\ref{eq71}) can be equivalently formulated as
\begin{equation}\label{eq72}
\begin{gathered}
  \mathop {\min }\limits_{\mathbf{p}} \sum\limits_{k = 1}^K {\sum\limits_{m = 1}^4 {\frac{{{\left|\tilde b_{1m,k}\right|^2+\left|\tilde b_{2m,k}\right|^2}}}{{{p_k}{\lambda _{m,k}} + 1}}} }  \hfill \\
  s.t.\;\;\sum\limits_{k = 1}^K {{{\log }_2}\left( {1 + {\rho _k}{p_k}} \right)}  \ge R_t, \sum\limits_{k = 1}^K {{p_k}}  = {P_T}, p_k \ge 0, \forall k. \hfill \\
\end{gathered}
\end{equation}
\begin{lemma}
$\lambda_{m,k} \ge 0, \forall m, \forall k$.
\end{lemma}
\renewcommand{\qedsymbol}{$\blacksquare$}
\begin{proof}
Since $\lambda_{m,k}, \forall m$ are obtained as the eigenvalues of (\ref{eq67}), it is sufficient to prove that (\ref{eq67}) is a semidefinite matrix. Note that for any given $\mathbf{x} \in \mathbb{C}^{4 \times 1}$, it holds that
\begin{equation}\label{eq73}
  {{\mathbf{x}}^H}{{\mathbf{B}}^{ - 1/2}}{\mathbf{A}}{{\mathbf{B}}^{ - H/2}}{\mathbf{x}} = {\left( {{{\mathbf{B}}^{ - H/2}}{\mathbf{x}}} \right)^H}{\mathbf{A}}\left( {{{\mathbf{B}}^{ - H/2}}{\mathbf{x}}} \right) \ge 0,
\end{equation}
where the inequality holds based on the fact that $\mathbf{A} \succeq \mathbf{0}$. Therefore, ${{\mathbf{B}}^{ - 1/2}}{\mathbf{A}}{{\mathbf{B}}^{ - H/2}}$ is by definition a semidefinite matrix, which completes our proof.
\end{proof}
\begin{thm}
  Problem (\ref{eq72}) is convex.
\end{thm}
\begin{proof}
  It can be readily seen that the sum-rate constraint in (\ref{eq72}) is concave in $\mathbf{p}$, and the power constraints are linear. Therefore, it is sufficient to prove that the objective function of (\ref{eq72}) is convex. For each term in the summation, the second-order derivative can be calculated as
  \begin{equation}\label{eq74}
    \frac{{{\partial ^2}\frac{{{{\left| {{{\tilde b}_{1m,k}}} \right|}^2} + {{\left| {{{\tilde b}_{2m,k}}} \right|}^2}}}{{{p_k}{\lambda _{m,k}} + 1}}}}{{\partial p_k^2}} = \frac{{2\left( {{{\left| {{{\tilde b}_{1m,k}}} \right|}^2} + {{\left| {{{\tilde b}_{2m,k}}} \right|}^2}} \right)\lambda _{m,k}^2}}{{{{\left( {{p_k}{\lambda _{m,k}} + 1} \right)}^3}}} \ge 0,
  \end{equation}
where the inequality holds true based on Lemma 1, which implies that each term is a convex function in $p_k$. The summation is therefore convex. This completes the proof.
\end{proof}
Due to the convexity of (\ref{eq72}), one can efficiently solve it by using numerical tools, e.g., CVX.
\section{Numerical Results}
In this section, we present the numerical results to validate the effectiveness of the proposed techniques for both sensing and communication. Unless otherwise specified, both the RSU and the vehicles operate at $f_c = 30\;\text{GHz}$, and we use $\Delta T = 0.02\text{s}$ as the block duration, $\sigma^2 = {\sigma_C^2} = 1$ as the noise variances for radar and communication, and ${\tilde \alpha} = 1$ as the reference communication channel coefficient. For the state evolution noises, we set $\sigma_\theta = 0.02^\circ$, $\sigma_d = 0.2\text{m}$, $\sigma_v = 0.5\text{m/s}$ and $\sigma_\beta = 0.1$, respectively. Note that here the variances for the state evolution are small since they stand for the approximation errors in the evolution models, which are irrelevant to the actual SNR. Moreover, the difference between two adjacent states is small given the short time duration $\Delta T$. As a consequence, the state variances should be set small enough. For the measurement noise variance, we set $a_1 = 1$, $a_2 = 6.7 \times 10^{-7}$ and $a_3 = 2 \times 10^4$. Finally, the matched-filtering gain is assumed to be $G = 10$.
\subsection{Performance for Tracking A Single Vehicle}
We first study the communication and sensing performances of the proposed technique with respect to a single vehicle. Without loss of generality, the initial state of the vehicle is set to $\theta_0 = 9.2^\circ$, $d_0 = 25\text{m}$, $v_0 = 20\text{m/s}$ and $\beta_0 = 0.5 + 0.5j$. The antenna number at the vehicle is $M = 32$.
\begin{figure}[!t]
    \centering
    \includegraphics[width=0.85\columnwidth]{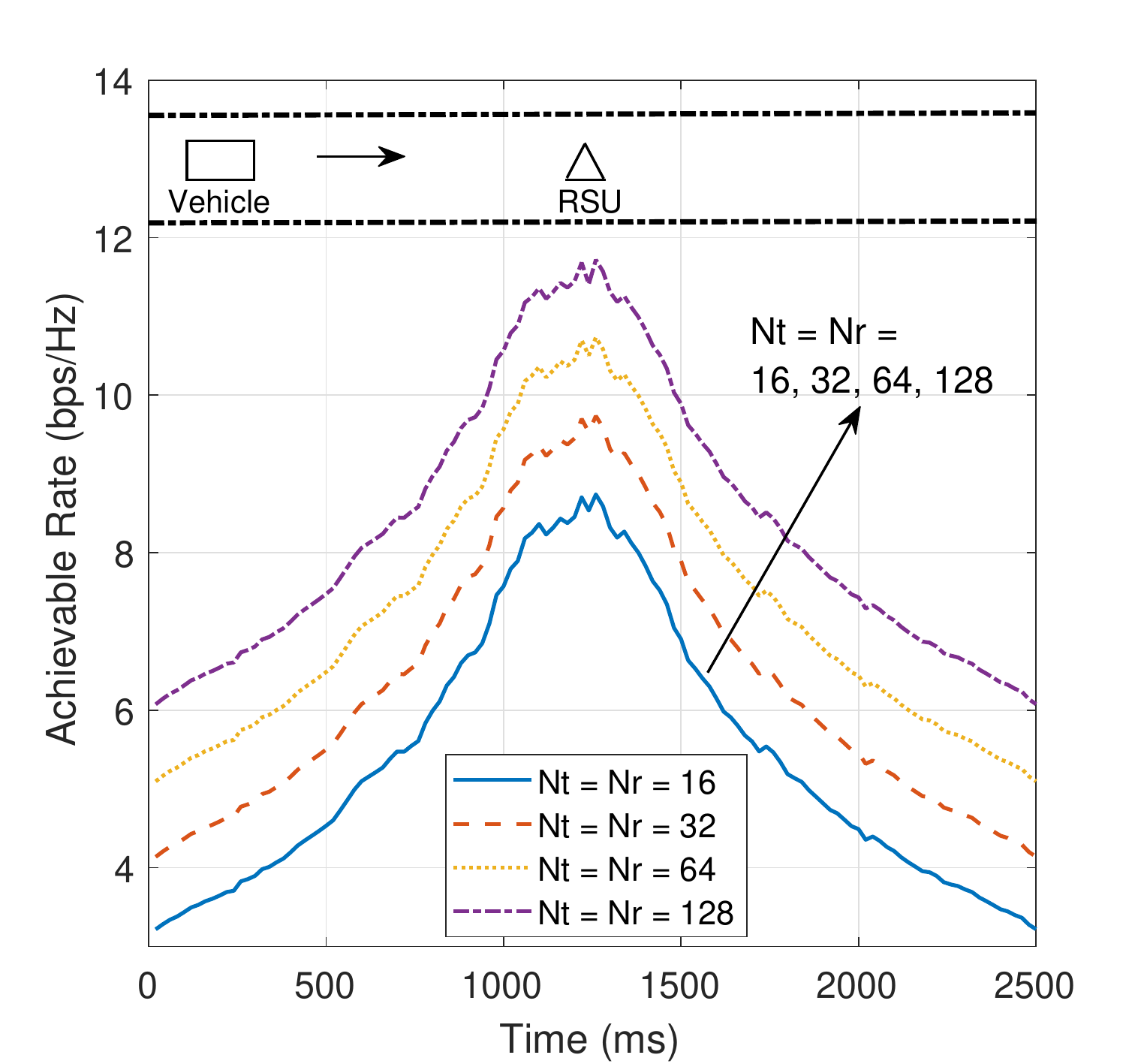}
    \caption{Achievable rate for a single vehicle, with initial state $\theta_0 = 9.2^\circ$, $d_0 = 25\text{m}$, $v_0 = 20\text{m/s}$ and $\beta_0 = 0.5 + 0.5j$, $M = 32$, $\text{SNR} = 10\text{dB}$.}
    \label{fig:4}
\end{figure}
\begin{figure}[!t]
    \centering
    \includegraphics[width=0.85\columnwidth]{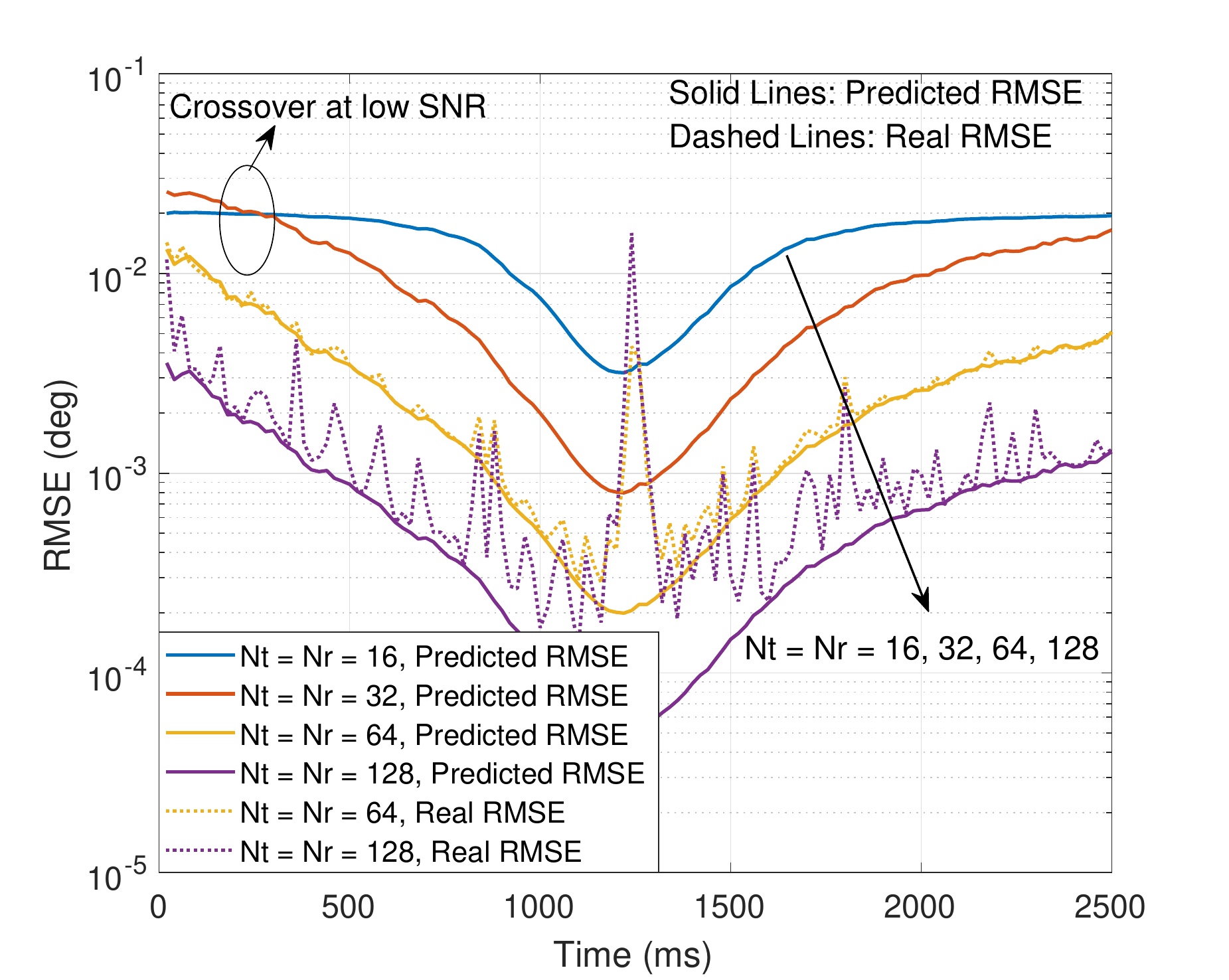}
    \caption{Angle tracking performance for a single vehicle, with initial state $\theta_0 = 9.2^\circ$, $d_0 = 25\text{m}$, $v_0 = 20\text{m/s}$ and $\beta_0 = 0.5 + 0.5j$, $M = 32$, $\text{SNR} = 10\text{dB}$.  }
    \label{fig:5}
\end{figure}
\begin{figure}[!t]
    \centering
    \includegraphics[width=0.85\columnwidth]{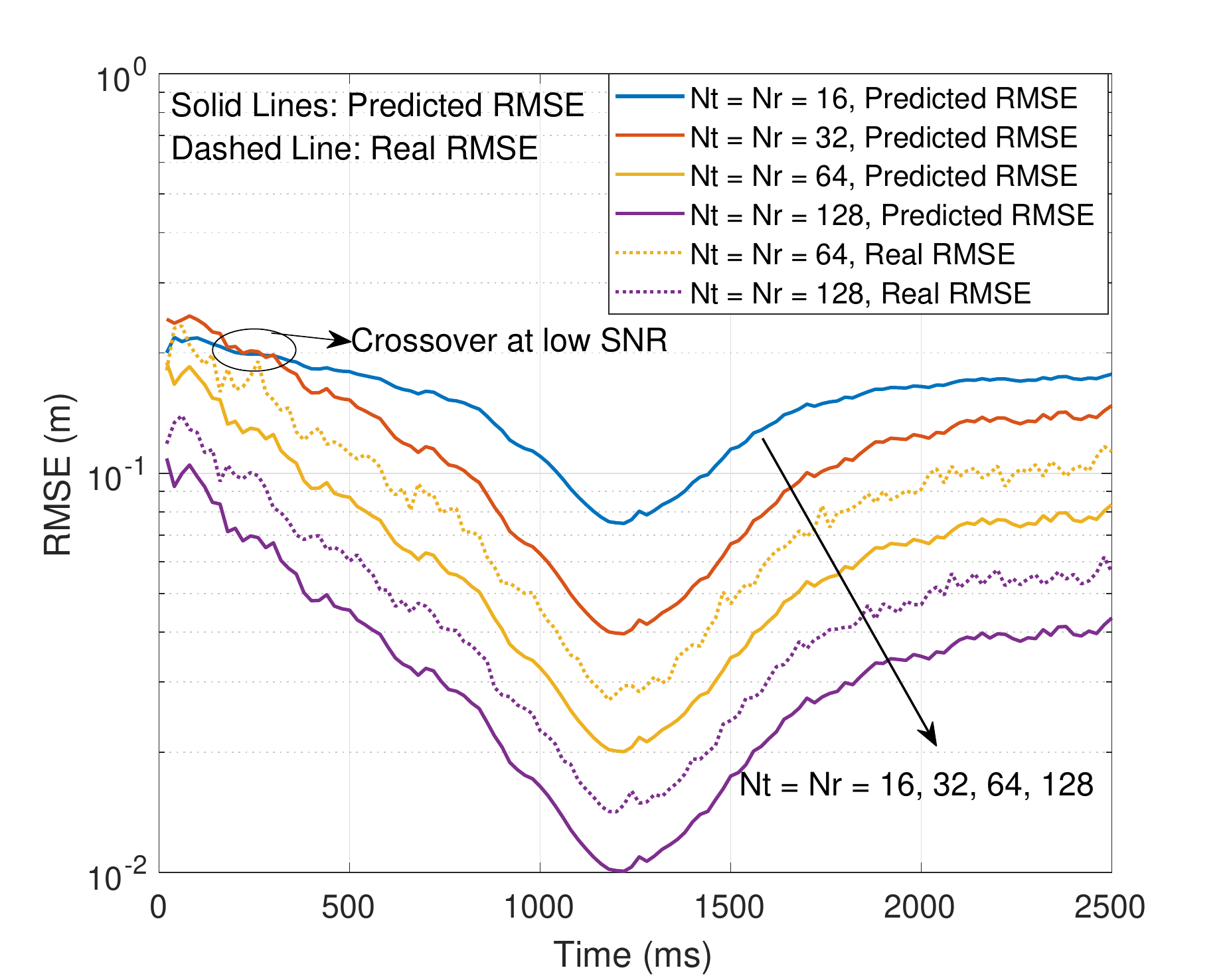}
    \caption{Distance tracking performance for a single vehicle, with initial state $\theta_0 = 9.2^\circ$, $d_0 = 25\text{m}$, $v_0 = 20\text{m/s}$ and $\beta_0 = 0.5 + 0.5j$, $M = 32$, $\text{SNR} = 10\text{dB}$.  }
    \label{fig:6}
\end{figure}
\\\indent In Fig. 4, we show the achievable downlink rates of the communication link with increasing antenna number at the RSU. In the considered scenario, the vehicle starts from one side of the RSU, then passes in front of the RSU to its other side. It can be observed that all the rates increase initially and then decrease with the evolution of time, which is caused by the change of the distance between the vehicle and the RSU. The maxima in the rates correspond to the positions where the vehicle is closest to the RSU, i.e., at the point where $\theta = 90^\circ$. Moreover, we see that by increasing both the transmit and receive antenna numbers at the RSU, the achievable rate is on the rise thanks to the improved array gain.
\\\indent In Fig. 5 and Fig. 6, we demonstrate the results for radar sensing in terms of root mean squared error (RMSE) for both angle and distance tracking. As expected, the general trend shows that the error reduces when the vehicle is approaching, and increases when it is driving away. When the vehicle is in front of the RSU, i.e., in the time region 1200 - 1400 ms when $\theta \approx 90^\circ$, the angle changes too fast for the EKF tracking, which causes the spikes for angel estimation. Fortunately, these error spikes will not affect the achievable rate significantly as shown in Fig. 4, despite that some small fluctuations may appear on the rate curves at the corresponding moments. One can see in both figures that the real RMSEs are tightly lower-bounded by the predicted RMSEs (which are also predicted PCRBs according to Theorem 1) in general \footnote{Note that the RMSE may exceed the predicted bound since the predicted RMSE does not equal to its real counterpart.}, which proves the correctness of our theoretical derivation in Sec. IV. In addition, it is interesting to note that crossover may occur for the 16- and 32-antenna curves in both Fig. 5 and Fig. 6. This is because when the vehicle is far away, i.e., when the receive SNR is low, large antenna array might be more likely to experience beam misalignment, since larger arrays will generate narrower beams. While such drawbacks can be compensated by the array gain, the 32-antenna array does not have sufficient gain to mitigate this effect, which hence performs worse than the 16-antenna case. For 64- and 128-antenna cases, the crossover vanishes thanks to the higher array gain.
\begin{figure}[!t]
\centering
\subfloat[]{\includegraphics[width=0.85\columnwidth]{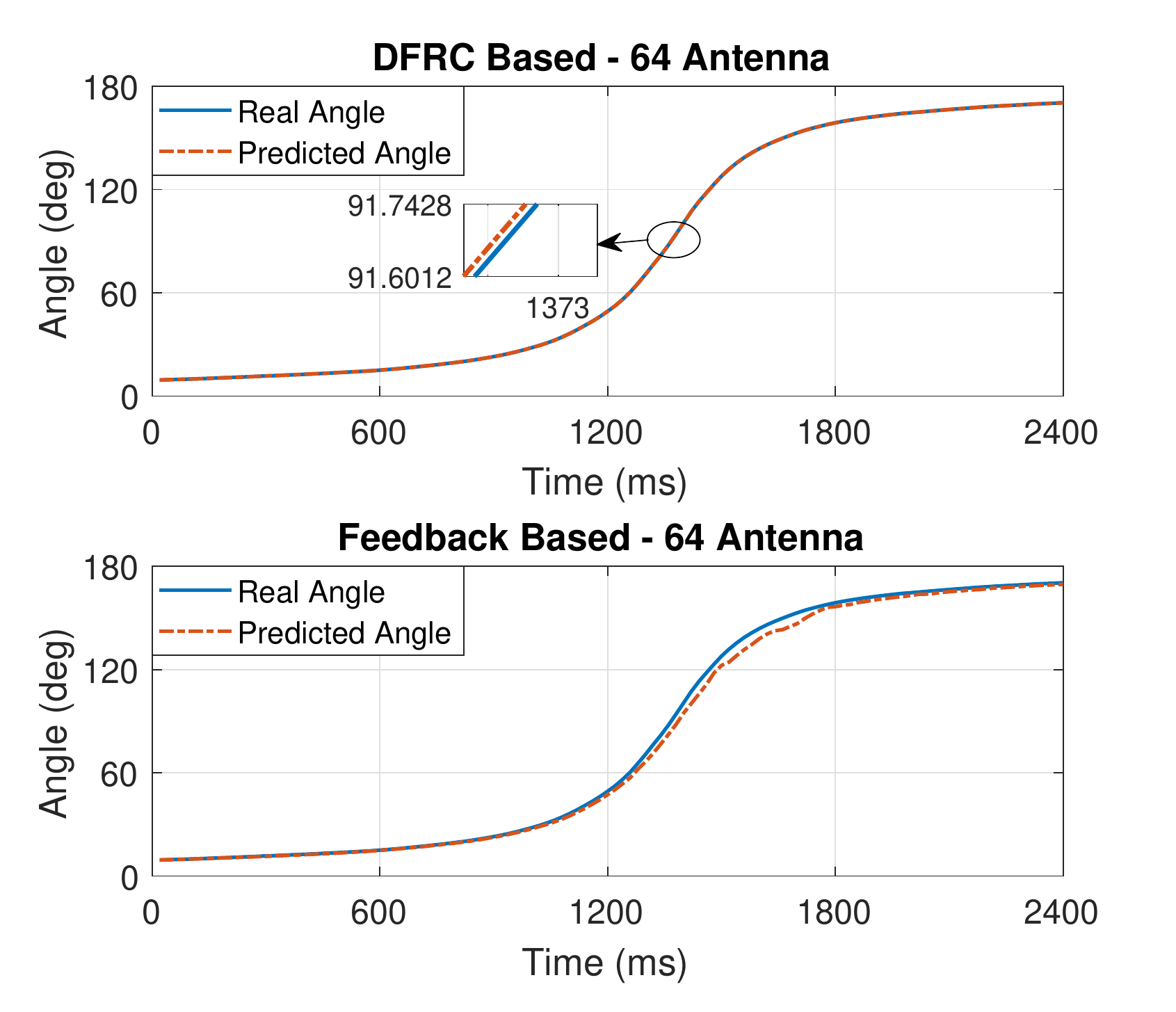}
\label{fig7a}}
\vspace{0.1in}
\hspace{.1in}
\subfloat[]{\includegraphics[width=0.85\columnwidth]{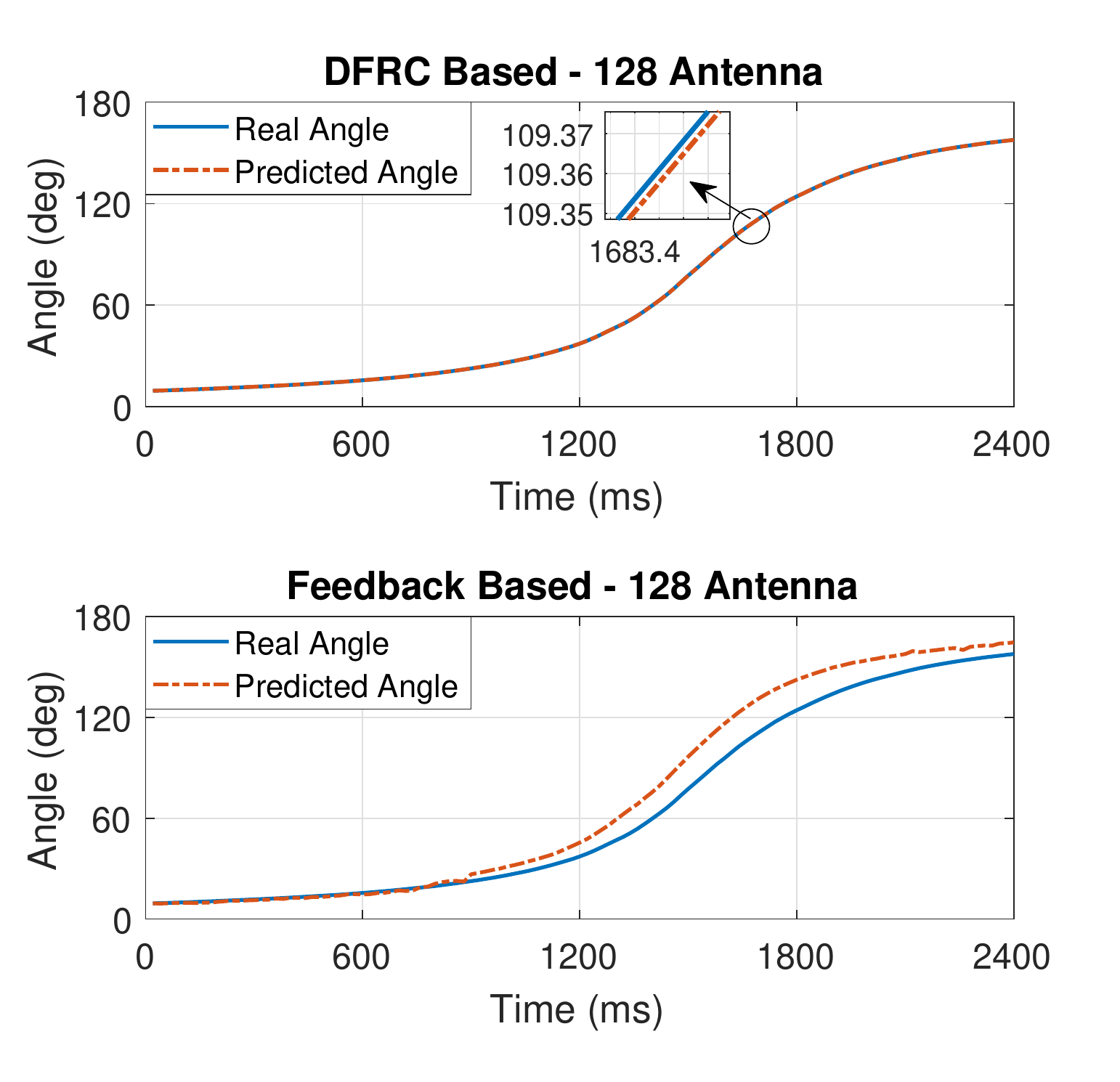}
\label{fig7b}}
\caption{Angle tracking performances for radar- and feedback-based schemes. (a) Initial state is set as $\theta_0 = 9.2^\circ$, $d_0 = 25\text{m}$, $v_0 = 18\text{m/s}$, ${\beta _0} = \frac{{\sqrt 2 }}{2} + \frac{{\sqrt 2 }}{2}j$ and $\alpha_0 = 25$, $N_t = N_r = M = 64$, $\text{SNR} = 10\text{dB}$; (b) Initial state is set as the same as (a), antenna number $N_t = N_r = M = 128$, and $\text{SNR} = 10\text{dB}$.}
\label{fig7}
\end{figure}
\begin{figure}[!t]
    \centering
    \includegraphics[width=0.85\columnwidth]{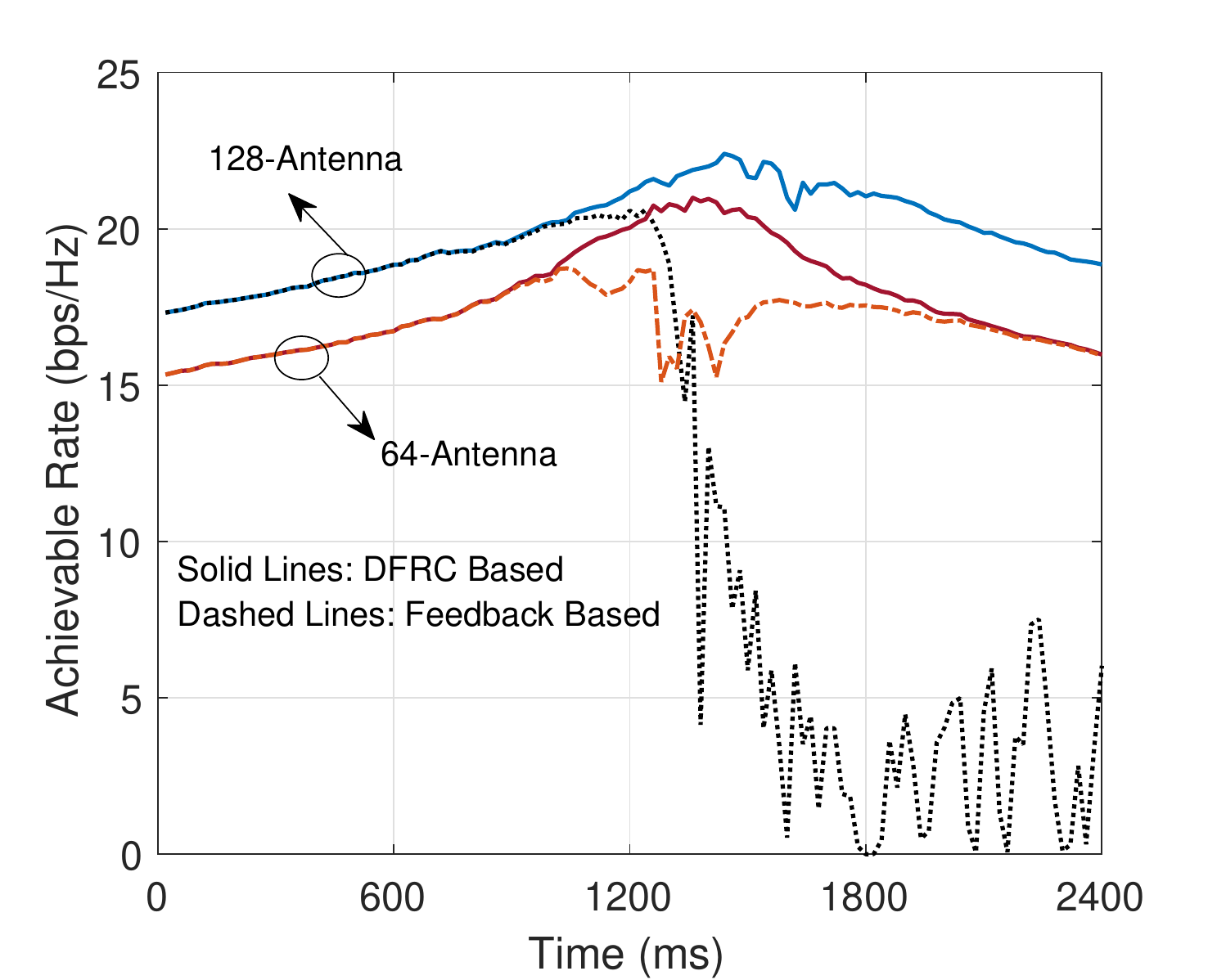}
    \caption{Achievable rate performances for radar- and feedback-based schemes, with initial state $\theta_0 = 9.2^\circ$, $d_0 = 25\text{m}$, $v_0 = 18\text{m/s}$, ${\beta _0} = \frac{{\sqrt 2 }}{2} + \frac{{\sqrt 2 }}{2}j$ and ${\tilde \alpha} = 25$, antenna number $N_t = N_r = M = 64$, and $\text{SNR} = 10\text{dB}$.  }
    \label{fig:8}
\end{figure}
\begin{table*}
\centering
\caption{Initial States for the Multi-Vehicle Scenario}
\begin{tabular}{@{}l|ccccc@{}}
\toprule
 & Vehicle 1 & Vehicle 2 & Vehicle 3 & Vehicle 4 & Vehicle 5 \\ \midrule
Angle & $7.66^\circ$ & $6.56^\circ$ & $5.74^\circ$ & $5.10^\circ$ & $4.59^\circ$ \\ 
Distance & 30m & 35m & 40m & 45m & 50m \\ 
Velocity & 20m/s & 18m/s & 16m/s & 12m/s & 10m/s \\ 
Reflection Coefficient & $2+2j$ & $1+j$ & $0.5+0.5j$ & $0.3+0.3j$ & $0.2+0.2j$ \\ \bottomrule
\end{tabular}
\end{table*}
\subsection{Performance Comparison for DFRC and Benchmark Feedback-Based Schemes}
In this subsection, we compare the performance of the proposed DFRC-based beam tracking scheme and the benchmark communication-only feedback-based method. In conventional EKF based beam tracking schemes such as \cite{7905941,8851151,8830375}, the transmitter sends a single pilot vector to the receiver at each epoch. The receiver then combines the pilot by a receive beamformer, estimates the angle and feeds it back to the transmitter, which is used for predicting the transmit beam of the next time-slot. Note that all the existing EKF based beam tracking schemes employ state evolution models that are different to that of our paper \cite{7905941,8851151,8830375}. For the sake of fairness, we use the same state evolution model in the feedback based scheme for comparison, except that the reflection coefficient $\beta$ is now replaced by the LoS channel coefficient $\alpha$, which is assumed to be perfectly known. The angle measurement model for the feedback scheme is based on (\ref{eq11}), with the pilot symbol being matched-filtered. The distance and the velocity measurements are based on the time delay and the Doppler shift as well. Since there is only one single pilot being employed for tracking in the feedback based method, the matched-filtering gain is $G = 1$ in contrast to the DFRC case where $G = 10$. As a consequence, the measurement variances for the feedback-based approach are at an order of magnitude that is 10 times of that of the proposed DFRC scheme. The initial state of the vehicle is set as $\theta_0 = 9.2^\circ$, $d_0 = 25\text{m}$, $v_0 = 18\text{m/s}$, ${\beta _0} = \frac{{\sqrt 2 }}{2} + \frac{{\sqrt 2 }}{2}j$ and ${\tilde \alpha} = 25$. Note that here we set ${\tilde \alpha} = d_0$ such that the modulus of the initial channel coefficient $\alpha_0$ is 1, which is the same as the reflection coefficient $\beta_0$ used in the DFRC scheme. Finally, we assume $N_t = N_r = M$ for fair comparisons.
\\\indent We first look at the angle tracking performance in Fig. 7(a) and (b). Thanks to the matched-filtering gain in the DFRC based technique, both of the figures reveal that the proposed scheme can accurately track the variation of the vehicle's angle, while the feedback-based scheme shows larger tracking errors. Another important reason for this is that the feedback based approach requires the vehicle to use a receive beamformer to combine the pilot signal, which projects the pilot signal to a lower-dimensional space. This inevitably causes the loss of the angular information. On the other hand, the proposed DFRC scheme does not perform receive beamforming for the reflected echoes, which preserves most of the angular information. It is interesting to see from Fig. 7(b) that when the antenna array has larger size, the tracking error for the feedback based approach goes up, since the beam becomes narrower and the added SNR gain is not sufficient.
\\\indent In Fig. 8, we show the achievable rates for both DFRC and feedback based techniques. It can be observed that at the beginning when the angle variation is relatively slow, the feedback technique leads to almost the same rate performance as that of the proposed method. When the vehicle is approaching the RSU, however, the angle begins to vary rapidly, and the rate of the feedback method decreases drastically, which is consistent with the associated angle tracking performance in Fig. 7. Moreover, it can be observed in the 64-antenna case that the rate of the feedback based method catches up with that of the DFRC method when the vehicle is driving away. This is because when the angular variation is slow, the angle tracking error becomes acceptable, and the EKF is still able to correct the angle deviation as shown in Fig. 7(a). Nevertheless, it fails to do so in the 128-antenna scenario given the narrower beam and higher misalignment probability, as depicted in Fig. 7(b).
\subsection{Performance for Tracking Multiple Vehicles}
In Figs. 9-11, we evaluate the performance of the proposed power allocation algorithm for joint sensing and communication of multiple vehicles. Without loss of generality, we consider 5 vehicles driving towards the same direction, with their initial states being set as TABLE. I. The transmit SNR in the multi-vehicle scenario is redefined as $\frac{P_T}{\sigma^2}$, where $P_T$ is the total power budget. The antenna numbers are set as $N_t = N_r = 128$ for the RSU, and $M = 32$ for the vehicles. In order to set the rate threshold $R_t$ in the problem (\ref{eq72}), we use the water-filling rate under the given power budget $P_T$ as a reference, which we denote as $R_{\max}$. This rate is obtained by using (\ref{eq_wf_solution}), and is the upper-bound of all the achievable sum-rates. Accordingly, we use $R_t = 0.9R_{\max}$ throughout the following results, and compare the performance of the proposed scheme with that of the water-filling power allocation under the same EKF beam tracking framework.
\begin{figure}[!t]
    \centering
    \includegraphics[width=0.85\columnwidth]{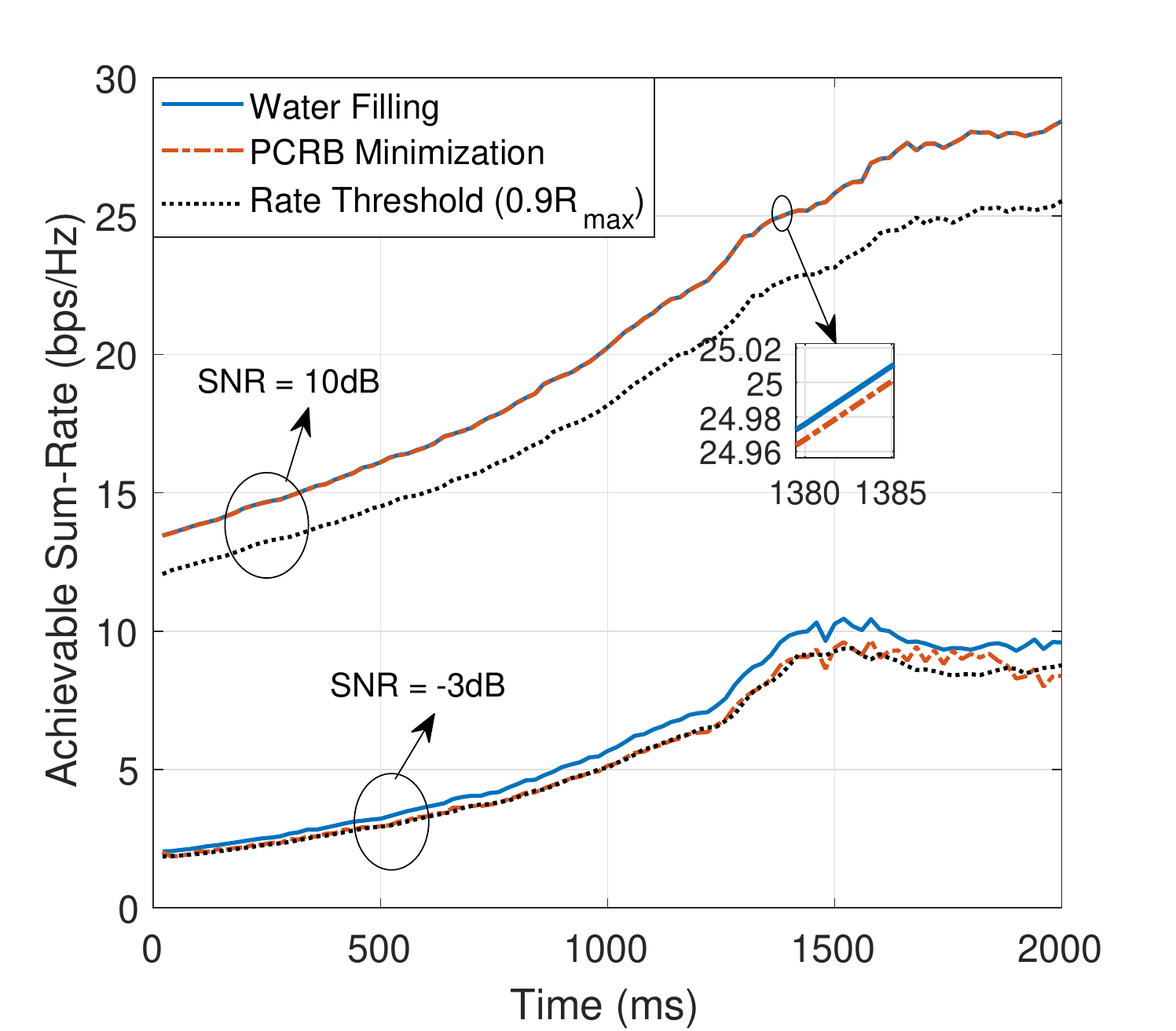}
    \caption{Achievable sum-rate for multiple vehicles. Initial states are configured as TABLE I. $N_t = N_r = 128$, $M = 32$, $R_t = 0.9R_{\max}$.}
    \label{fig:9}
\end{figure}
\\\indent We firstly show the sum-rates obtained for both the water-filling and the proposed PA in Fig. 9 under $\text{SNR} = -3\text{dB}$ and $10\text{dB}$, respectively. It can be seen that when the SNR is low, the rate constraint for the PA problem (\ref{eq72}) is tight, i.e., the resultant sum-rate equals to the threshold in general. Note here that the sum-rate might exceed the limit of the threshold $R_t$, since $R_t$ is set based on the predicted angles, where the beamforming gain is ideally set as 1. Nevertheless, the actual sum-rate is calculated based on the real angles in the real channel, in which case the beamforming gain will always be lower than 1. On the other hand, when the SNR is high, the rate constraint is no longer activated, i.e., the obtained sum-rate is always greater than the threshold, which shows almost the same performance with its water-filling counterpart.
\begin{figure}[!t]
    \centering
    \includegraphics[width=0.85\columnwidth]{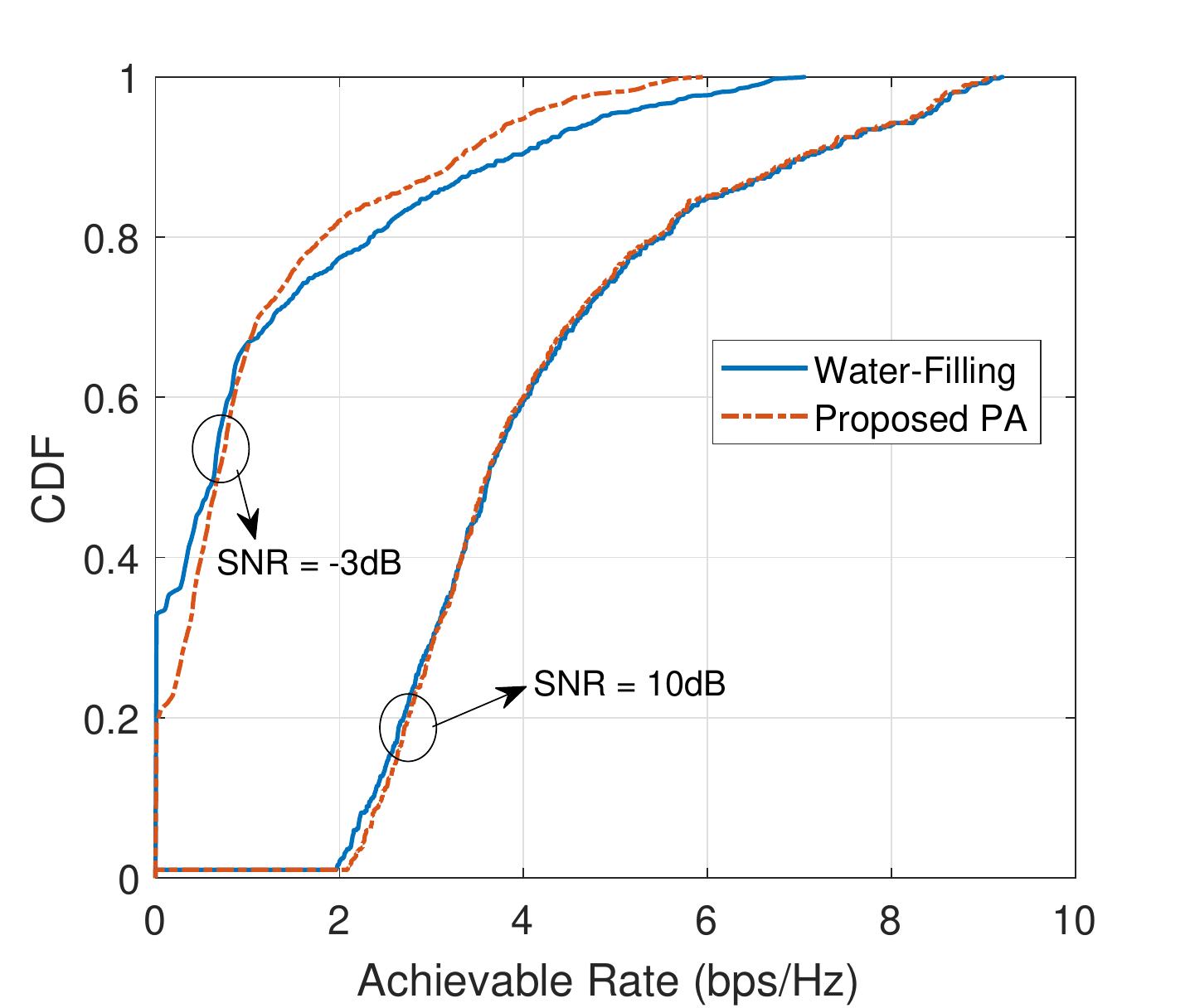}
    \caption{CDF of the achievable rates for multiple vehicles. Initial states are configured as TABLE I. $N_t = N_r = 128$, $M = 32$, $R_t = 0.9R_{\max}$.}
    \label{fig:10}
\end{figure}
\begin{figure}[!t]
    \centering
    \includegraphics[width=0.85\columnwidth]{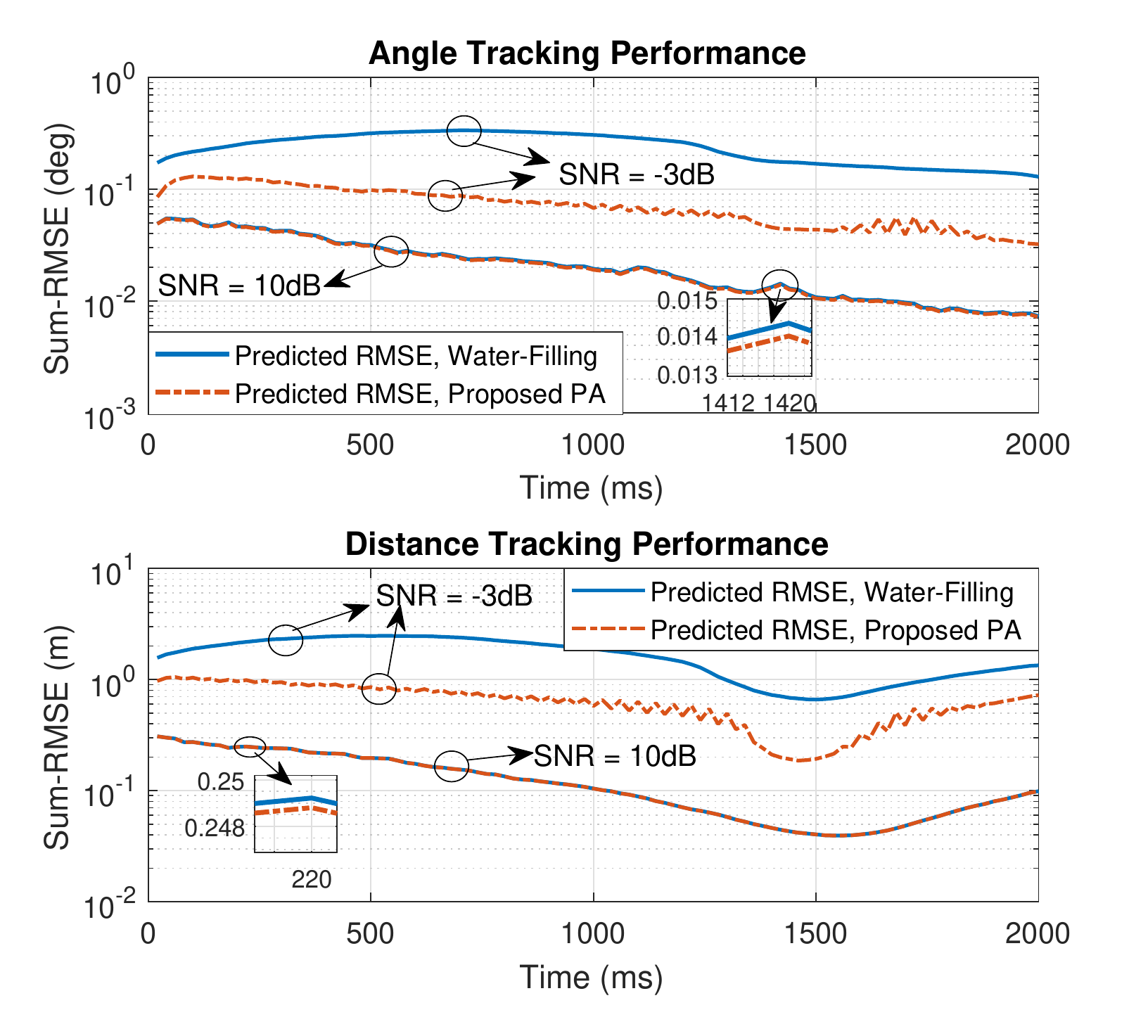}
    \caption{Angle and distance tracking performance for multiple vehicles. Initial states are configured as TABLE I. $N_t = N_r = 128$, $M = 32$, $R_t = 0.9R_{\max}$.}
    \label{fig:11}
\end{figure}
\\\indent In Fig. 10, we explore the rate fairness performance by showing the empirical cumulative distribution functions (CDF) in terms of the achievable rates for each vehicle. One notices that while the water-filling PA might have high probability to generate higher sum-rate, the proposed PA addresses fairness in the sense of ensuring the minimum achievable rate. This has been evidently shown in both SNR regimes, where crossovers between two CDF curves appear at $1\text{bps/Hz}$ and $3.5\text{bps/Hz}$ for $\text{SNR} = -3\text{dB}$ and $10\text{dB}$, respectively. This is because the water-filling scheme is more likely to allocate small or even zero power to those vehicles with bad channel quality. As a result, the proposed PA technique achieves more fairness among all the vehicles.
\\\indent Finally, in Fig. 11, we investigate the sensing performance of the proposed PA in terms of the predicted RMSE, and again, with the water-filling approach as a benchmark. As expected, at the high-SNR regime, both techniques show similar sensing performances. In fact, we observe in our simulation that both of the proposed algorithm and the water-filling scheme are likely to generate a uniform power allocation when the SNR is high. It is worth highlighting that at the low-SNR regime, the proposed PA scheme considerably outperforms its water-filling counterpart on both angle and distance tracking. This is because the water-filling PA is not tailored for sensing, while the proposed PA minimizes the summation of the PCRB of all the vehicles and guarantees the rate constraint at the same time. In other words, the proposed method is able to trade off between the communication and the sensing performances by sacrificing the communication rate for improved estimation accuracy. A remarkable conclusion that can be learned from the above results is that under the considered PA scenario, \emph{higher estimation accuracy does not necessarily mean higher sum-rate}. By using the proposed PA, the RMSE of the angle has been reduced to the level of $0.01^\circ$, which improves the beamforming gain. Nonetheless, this power allocation is tailored for minimizing the sum-PCRB, which is different from that of the optimal water-filling method, and may thus lead to the loss in the SNR of the vehicles with good channel conditions. On the other hand, although the water-filling PA might have larger estimation errors on the angle, the power allocated to the vehicles can still compensate for the loss in the beamforming gain.

\section{Conclusion}
In this paper, we have proposed a novel predictive beamforming design for the vehicular network by utilizing joint radar and communication functionalities deployed on the road side unit (RSU). Aiming for tracking the angular variation of the vehicle, we have proposed an extended Kalman filtering (EKF) framework that builds upon the measurements of the echo signal as well as the state evolution model of the vehicle. In order to minimize the tracking errors for multiple vehicles in the Kalman iteration while guaranteeing the downlink quality-of-service (QoS), we have proposed an optimization-based power allocation scheme that is capable of minimizing the posterior Cram\'er-Rao bound (PCRB) for both angle and distance estimation under given sum-rate threshold and power budget. Numerical results have been provided to validate the proposed techniques, which have shown that the dual-functional radar-communication (DFRC) based beamforming design significantly outperforms the communication-only feedback-based schemes, and that the proposed power allocation algorithm yields the minimum tracking error while achieving a favorable trade-off between sensing and communication performances.

%




\ifCLASSOPTIONcaptionsoff
  \newpage
\fi



\bibliographystyle{IEEEtran}
\bibliography{IEEEabrv,Veh_REF}
\end{document}